\documentclass[aps,prd,reprint,superscriptaddress,nofootinbib,longbibliography]{revtex4-2}

\usepackage[T1]{fontenc}
\usepackage[utf8]{inputenc}
\usepackage{lmodern}
\usepackage{amsmath,amssymb,mathtools,amsthm,bm,booktabs}
\usepackage[colorlinks=true,linkcolor=blue,citecolor=blue,urlcolor=blue]{hyperref}
\usepackage[nameinlink,capitalise,noabbrev]{cleveref}
\usepackage{lipsum}
\usepackage{graphicx}

\newtheorem{proposition}{Proposition}

\newcommand{\TV}{d_{\mathrm{TV}}}
\newcommand{\KL}{D_{\mathrm{KL}}}
\newcommand{\I}{I}
\newcommand{\E}{\mathbb{E}}
\newcommand{\Prob}{\mathbb{P}}
\newcommand{\1}{\mathbf{1}}

\begin{document}

\title{Loophole-Robust Certification of Quantum Advantage}

\author{Prosanta Pal}
\affiliation{Department of Physics and Astronomy, Clemson University, Clemson SC, USA 29634}

\author{Gargee Sharma}
\affiliation{Indian Institute of Technology Delhi, Hauz Khas, New Delhi 110016, India}

\author{Ramakrishna Podila}
\email{rpodila@g.clemson.edu}
\affiliation{Department of Physics and Astronomy, Clemson University, Clemson SC, USA 29634}

\date{\today}

\begin{abstract}
Claims of quantum advantage should remain robust even when classical strategies have access to side information correlated with the benchmark under evaluation, just as Bell certification must account for measurement dependence. We formalize such correlations as benchmark dependence, a task-level generalization of measurement dependence. For every bounded-reward task, we show that the
optimal benchmark-dependent classical score obeys
\(S_\eta\leq\min\{1,S_{\mathrm{cl}}+\eta\}\), and construct a family of tasks that saturates this bound, showing that the linear dependence on \(\eta\) is tight without further assumptions. For repeated product tasks with roundwise dependence, we obtain the stronger multiplicative bound
\(S_\eta^{(n)}\leq(\omega_{\mathrm c}+\eta)^n\), and extend the framework to finite-sample data, mutual-information constraints, multipartite tasks, and correlations distributed along a causal path. Applying these results to aggregated IBM hardware data, we obtain positive raw-count cycle-product certificates of \(0.0812\) for CHSH and \(0.2178\) for
Mermin--GHZ, while the nine-context magic-square construction remains
uncertified; readout-mitigated values are reported separately as
sensitivity estimates. We also analyze a non-Bell quantum-kernel benchmark, where a label-construction variable has measured conditional dependence
\(\widehat{\eta}_{\lambda}^{(Y)}=0.5\), above the threshold
\(\eta_{\mathrm{req}}=0.375\), required to close the reported score gap, and yields perfect classical classification. The framework therefore converts a quantum--classical score separation into a quantitative lower bound on the benchmark-correlated classical information required to explain the score separation.
\end{abstract}

\maketitle

\section{Introduction}

Quantum advantage denotes a demonstrable separation between quantum and classical
strategies for the same task under a common resource model and success
criterion. In complexity-theoretic settings, this separation is asymptotic:
a family of quantum algorithms requires fewer resources than any classical
algorithm achieving comparable accuracy and success probability
\cite{HarrowMontanaro2017,AaronsonChen2017, preskill2012quantum}. The resource may be
runtime, memory, communication, query complexity, or sample complexity. By contrast, 
experimental demonstrations establish finite-instance separations
for a specific device, benchmark distribution, and verification metric, as in
random-circuit sampling or photon sampling \cite{Arute2019,zhong2020quantum,wu2021strong,madsen2022quantum}. Such claims inherit the assumptions underlying the benchmark’s classical hardness and remain sensitive to advances in classical simulation and verification \cite{huang2020classical}.
Practical quantum advantage imposes
the stronger requirement that the comparison include end-to-end costs and
the strongest available classical methods
\cite{DaleyEtAl2022,HibatAllah2024}. A quantum-advantage claim is therefore never absolute; it is defined relative to precise quantum and classical strategy classes and, crucially, to the information and side resources available to each.

Many problems in quantum information admit a direct score-based formulation~\cite{helstrom1969quantum,Cleve2004,buhrman2009non,Brunner2014}.
A benchmark instance is sampled from a prescribed distribution, a strategy
produces an output, and a bounded reward quantifies its performance. The
classical and quantum values are the optimal expected rewards over the
corresponding strategy classes. Nonlocal games provide the canonical
example: classical players are restricted to local response functions
assisted by shared randomness, whereas quantum players may additionally
share entanglement \cite{Brunner2014,Bell1964,Cleve2004,buhrman2009non}. Analogous quantum-classical performance gaps arise in quantum learning and distributed information processing, although the resource exhibiting the advantage differs across settings. For instance, quantum algorithms can provide provable computational speedups for specially constructed supervised-learning problems \cite{Liu2021}; quantum-enhanced measurement and coherent processing can reduce the number of experimental preparations required to infer properties of an unknown quantum system \cite{huang2020predicting,Huang2022}; and shared entanglement can reduce the communication or model resources required for certain distributed learning tasks \cite{buhrman1998quantum,ZhaoDeng2025}.
In each case, the attainable score depends on the computational and the data-access model, and any available auxiliary information. This is particularly consequential in quantum machine learning.
Quantum feature maps can define representations that are difficult to
evaluate classically \cite{Havlicek2019}, but computational hardness alone
does not imply predictive advantage. Training data may allow a classical
learner to approximate quantities that remain difficult to compute directly,
and the outcome of a quantum--classical comparison can change with the
assumed data-access model \cite{HuangBroughtonEtAl2021}. In general,
classical and quantum data access, feature representation, trainability, and generalization all enter the definition of the learning problem
\cite{CerezoEtAl2022}, motivating a reassessment of whether
outperforming a restricted set of classical models is, by itself, an
adequate criterion for quantum advantage \cite{SchuldKilloran2022}. 

Engineered benchmarks are particularly vulnerable to this ambiguity. If the labels are generated from a hidden quantity associated with the quantum representation, that same quantity may also give a classical model an easy route to the answer. This is related to data leakage and shortcut learning in classical machine learning, where a model exploits information or correlations that are not part of the intended task \cite{KaufmanEtAl2012,GeirhosEtAl2020}. The analogy is, however, not exact: benchmark-construction information need not be present
in the declared training features to constitute a classical side resource.
A score gap relative to the official inputs may therefore disappear once
such construction-side information is made available.

Bell nonlocality provides the established framework for quantifying such
assumption dependence. Bell's theorem excludes local hidden-variable models
under explicit causal and statistical assumptions \cite{Bell1964}, while
the Clauser--Horne--Shimony--Holt (CHSH) inequality converts this
incompatibility into an experimentally accessible bound \cite{CHSH1969}.
Loophole-free experiments demonstrated violations while controlling the
principal experimental routes to a local explanation
\cite{Hensen2015,Giustina2015,Shalm2015}. Relaxed Bell inequalities go
further by determining how the classical bound changes when an assumption
is weakened. In measurement-dependent models, the hidden variable may be
correlated with the measurement settings, and the strength of that
correlation becomes a quantifiable classical resource
\cite{Hall2010,Hall2011,ThinhSheridanScarani2013,
PutzRossetBarneaLiangGisin2014,PutzGisin2016,Friedman2019}. A Bell
violation can therefore be supplemented by a robustness statement: the
minimum measurement dependence required for a classical model to reproduce
the observed correlations.
The Bell-test perspective suggests a broader question: can the robustness
of a claimed quantum advantage be quantified when a classical strategy has
access to side information correlated with the evaluated benchmark?
Outside Bell scenarios, the relevant correlation may involve a query, test
example, distributed input, benchmark-selection rule, or latent
construction variable. Despite their different physical interpretations,
these cases share a common operational structure. A classical strategy is
evaluated on an instance \(z\), while its hidden state or side resource
\(\lambda\) is ordinarily assumed to be independent of \(z\). Relaxing
this assumption can increase the attainable classical score, and the
central certification question is how much dependence is required to
explain an observed quantum--classical separation.

Motivated by this correspondence, we introduce \emph{benchmark
dependence}, the task-level analogue of measurement dependence \cite{Hall2010}. For a
benchmark distribution \(q(z)\), let \(p(\lambda|z)\) denote the
hidden-variable distribution available to the classical strategy for
instance \(z\), and define the benchmark marginal
\(
p_0(\lambda)=\sum_z q(z)p(\lambda|z).
\)
We quantify benchmark dependence by
\(
\eta_{\mathrm{bm}}
=
\sup_z
d_{\mathrm{TV}}
\!\left(
p(\lambda|z),p_0(\lambda)
\right).
\)
The benchmark-independent model corresponds to
\(\eta_{\mathrm{bm}}=0\). When \(z\) is the full measurement-setting
tuple, this definition reduces to a posterior total-variation measure of
measurement dependence. For a general benchmark,
\(\eta_{\mathrm{bm}}\) quantifies the strength of any classical side
resource whose distribution depends on the evaluated instance.

Our central result is a universal relaxed advantage inequality. For any
distributed task with reward in \([0,1]\), the optimal classical score under
benchmark dependence of strength at most \(\eta\) satisfies
\(
S_{\eta}
\leq
\min\!\left\{1,S_{\mathrm{cl}}+\eta\right\},
\)
where \(S_{\mathrm{cl}}\) is the optimal benchmark-independent classical
value. Hence, an observed score \(S>S_{\mathrm{cl}}\) implies
\(
\eta_{\mathrm{bm}}
\geq
S-S_{\mathrm{cl}}
\)
for every benchmark-dependent classical surrogate capable of reproducing
that score. The bound is independent of the number of parties and of the
sizes of the input, output, and hidden-variable spaces. We also construct a family of tasks that saturates the inequality, proving
that the linear dependence on \(\eta\) is universally optimal. Benchmark
dependence can therefore translate directly into an equal increase in the
attainable classical score.

The variational estimate used in proving the additive bound is the standard
total-variation control of the expectation of a bounded reward, familiar
from robust statistics and statistical decision theory
\cite{Huber1964,LeCam1964,LeCamYang2000}. Our contribution is not this
probabilistic inequality in isolation, but its formulation as an
operational resource bound for quantum-advantage certification: the
dependence is assigned to a benchmark-correlated classical side resource,
the score is optimized over the resulting distributed strategy class, and
the bound is inverted to quantify the minimum resource required by a
classical explanation. The exact saturation result further shows that the
resulting unit coefficient cannot be improved for general bounded-reward
tasks.

This result establishes a common framework for relaxed Bell inequalities
and the broader certification of quantum advantage. For Bell-derived tasks,
it quantifies the measurement dependence required by a classical
explanation. For non-Bell tasks, it provides the same operational measure
for benchmark leakage and construction-side information. We complement the
universal bound with finite-sample certification, an average-case
mutual-information formulation, extensions to correlated and multipartite
settings, and explicit benchmark-dependent classical surrogates. We apply
the framework to CHSH, the Greenberger--Horne--Zeilinger
(GHZ)/Mermin game, and the Mermin--Peres magic-square game
\cite{Mermin1990,Peres1990}, using data from IBM Fez, Kingston, and
Marrakesh, and to a Havlíček-style quantum-kernel benchmark
\cite{Havlicek2019} in which construction-side variables can be directly
audited.

The framework separates two distinct questions. Improving the classical
algorithm changes the benchmark-independent value \(S_{\mathrm{cl}}\);
allowing information correlated with the evaluated instance changes the
classical resource model itself. Relaxed advantage inequalities quantify the
latter and convert an observed quantum--classical score gap into a
resource-sensitive robustness certificate: the minimum benchmark dependence
required for a classical explanation. In this way, we establish a task-level extension of the Bell-test loophole
calculus to the certification of quantum advantage in distributed information
processing and quantum machine learning.
\section{Methods}
\subsection{Framework}
This section formalizes quantum advantage in terms of task performance.
Let \(S\) be an observed or target quantum score for a distributed task,
and let \(S_{\mathrm{cl}}\) be the optimal benchmark-independent classical
value. Advantage corresponds to \(S>S_{\mathrm{cl}}\). Our goal is to
quantify the minimum benchmark-dependence strength \(\eta\) for which a
classical model can attain \(S_\eta\geq S\).

For each round of the task, let \(z\in\mathcal Z\) denote the benchmark instance, drawn from a
known distribution \(q(z)\). The instance \(z\) specifies the local inputs
\(x(z)\in\mathcal X\) and \(y(z)\in\mathcal Y\) provided to Alice and Bob, respectively. Alice
and Bob return outputs \(a\in\mathcal A\) and \(b\in\mathcal B\). Their performance on that
round is quantified by a reward function
\[
r:\mathcal Z\times\mathcal A\times\mathcal B \to [0,1],
\qquad
(z,a,b)\mapsto r(z,a,b),
\]
where \(r(z,a,b)=1\) corresponds to perfect success and \(r(z,a,b)=0\) to failure.

A strategy is specified by the conditional distribution \(p(a,b|z)\) of outputs given the
benchmark instance. Its average score is
\begin{equation}
S[p]:=\sum_{z\in\mathcal Z} q(z)\sum_{a\in\mathcal A}\sum_{b\in\mathcal B}
p(a,b|z)\,r(z,a,b).
\label{eq:score}
\end{equation}

We first consider loophole-free classical shared-randomness models. In this case, Alice and Bob
have access to a shared hidden variable \(\lambda\), distributed according to
\(p_0(\lambda)\), but \(\lambda\) is independent of the benchmark instance \(z\). Their local
response functions are \(p(a|x,\lambda)\) and \(p(b|y,\lambda)\), so the induced output
distribution is
\begin{equation}
p(a,b|z)=\sum_{\lambda} p_0(\lambda)\,p(a|x(z),\lambda)\,p(b|y(z),\lambda).
\label{eq:classical_ideal}
\end{equation}
The corresponding optimal loophole-free classical score is
\begin{equation}
S_{\mathrm{cl}}
:=
\sup_{p_0,\,p(a|x,\lambda),\,p(b|y,\lambda)} S[p].
\label{eq:Scl}
\end{equation}

To model benchmark dependence, we relax the assumption that the hidden-variable distribution is
independent of \(z\). We therefore allow
\begin{equation}
p(a,b|z)=\sum_{\lambda} p(\lambda|z)\,p(a|x(z),\lambda)\,p(b|y(z),\lambda).
\label{eq:classical_dep}
\end{equation}
As a reference, we define the benchmark marginal of the hidden variable by
\begin{equation}
p_0(\lambda):=\sum_{z\in\mathcal Z} q(z)\,p(\lambda|z).
\label{eq:benchmark_marginal}
\end{equation}
This leads to the benchmark-dependence parameter
\begin{equation}
\eta_{\mathrm{bm}}
:=
\sup_{z\in\mathcal Z}\TV\!\bigl(p(\lambda|z),p_0(\lambda)\bigr),
\label{eq:eta_bm}
\end{equation}
where the total-variation distance between two probability distribution $P(\lambda)$ and $Q(\lambda)$ is defined as
\begin{equation}
\TV(P,Q):=\frac12\sum_\lambda |P(\lambda)-Q(\lambda)|.
\label{eq:tv_def}
\end{equation}
We write \(S_\eta\) for the optimal score over all benchmark-dependent classical models of the
form \eqref{eq:classical_dep} satisfying \(\eta_{\mathrm{bm}}\le \eta\). Throughout, \(S_{\mathrm{cl}}\) denotes the benchmark-independent
classical optimum and \(S_\eta\) denotes the classical optimum under the
resource constraint \(\eta_{\mathrm{bm}}\leq\eta\). We reserve
\(\eta_{\mathrm{req}}\) for the dependence required to reproduce a
specified point-estimate score, and
\(\eta_{\min}^{\mathrm{cert}}(\alpha)\) for a finite-sample lower
confidence bound. A hat denotes an empirical quantity computed from a
finite benchmark, while a tilde denotes a readout-corrected point estimate
that is not assigned exact binomial coverage.

A second dependence measure, useful for comparison with the measurement-dependence literature, is
the pairwise parameter
\begin{equation}
\eta_{\mathrm{pair}}
:=
\sup_{z,z'\in\mathcal Z}\TV\!\bigl(p(\lambda|z),p(\lambda|z')\bigr).
\label{eq:eta_pair}
\end{equation}
These two quantities are equivalent up to constants:
\begin{equation}
\eta_{\mathrm{bm}}\le \eta_{\mathrm{pair}}\le 2\,\eta_{\mathrm{bm}}.
\label{eq:eta_compare}
\end{equation}
The left inequality follows because \(p_0\) is a convex mixture of the distributions
\(p(\lambda|z)\), while the right inequality is a direct consequence of the triangle inequality. 
\paragraph*{Relation to measurement dependence:}
In a Bell scenario, the benchmark instance \(z\) is simply the full measurement-setting tuple, so benchmark dependence becomes measurement dependence in the usual Bell sense. Our parameter \(\eta_{\mathrm{bm}}\) is a posterior worst-case measure, since it quantifies how much the hidden-variable law \(p(\lambda|z)\) can deviate from the benchmark marginal \(p_0(\lambda)\). It is therefore most directly comparable to total-variation formulations of measurement dependence such as those of Hall~\cite{Hall2010,Hall2011}. By contrast, the framework of measurement-dependent locality, min-entropy constraints, and Santha--Vazirani source assumptions are typically formulated as constraints on \(p(z|\lambda)\) rather than on \(p(\lambda|z)\)~\cite{PutzRossetBarneaLiangGisin2014,PutzGisin2016,ThinhSheridanScarani2013,SanthaVazirani1986}. The two descriptions are related by Bayes' rule once the benchmark distribution \(q(z)\) is fixed, but the translation is generally one-way rather than an equivalence. With Hall's pairwise \(L^1\) measurement-dependence parameter \(M\) ~\cite{Hall2010,Hall2011}, one has
\begin{equation}
\frac{M}{4}\le \eta_{\mathrm{bm}}\le \frac{M}{2},
\label{eq:Hall_translation}
\end{equation}
since \(M=2\eta_{\mathrm{pair}}\) and \(\eta_{\mathrm{bm}}\le \eta_{\mathrm{pair}}\le 2\eta_{\mathrm{bm}}\).

\paragraph*{Relation to statistical decision theory:}
The proof of the universal bound uses the standard fact that, for
\(0\leq f\leq 1\),
\begin{equation}
\left|
\sum_{\lambda}
\bigl[P(\lambda)-Q(\lambda)\bigr]f(\lambda)
\right|
\leq
d_{\mathrm{TV}}(P,Q).
\end{equation}
This stability of bounded risks under total-variation perturbations is
well established in robust statistics and statistical decision theory
\cite{Huber1964,LeCam1964,LeCamYang2000}. Here, however, the perturbation
acts on the conditional law of a classical side resource,
\(p(\lambda|z)\), rather than directly on the benchmark distribution.
Optimizing over the associated distributed response functions converts
the statistical inequality into a task-level bound on the classical value
and, upon inversion, into a robustness certificate for quantum advantage.

Finally, note that in the loophole-free model one has \(p(\lambda|z)=p_0(\lambda)\) for all
\(z\). Hence fixing \(p_0\) to the benchmark marginal does not change the definition of
\(S_{\mathrm{cl}}\): any benchmark-independent optimizer already coincides with its own
benchmark marginal.
\subsection{Universal relaxed inequality, exact saturation, and optimality}

\paragraph*{General relaxed advantage inequality:} Any classical model whose hidden variables are allowed to depend on the task instance can improve its score by at most \(\eta\) over the optimal ideal value. 

\begin{proof} 
Let
\begin{equation}
S_{\eta}
:=
\sup_{\substack{
p(\lambda|z),\,p(a|x,\lambda),\,p(b|y,\lambda)\\
\eta_{\mathrm{bm}}\le \eta
}}
S[p],
\label{eq:Seta_def}
\end{equation}
denote the optimal score achievable by any benchmark-dependent classical model whose dependence strength does not exceed \(\eta\). We prove that for every bounded-reward distributed task,
\begin{equation}
S_{\eta} \le S_{\mathrm{cl}}+\eta.
\label{eq:general_bound}
\end{equation}

Since all task scores lie in \([0,1]\), ~\eqref{eq:general_bound} may be written in the universally valid sharpened form
\begin{equation}
S_\eta \le \min\{1,\;S_{\mathrm{cl}}+\eta\}.
\end{equation}
This ceiling becomes active only for \(\eta \ge 1-S_{\mathrm{cl}}\). It therefore clarifies the saturation regime, especially for tasks with large classical benchmark \(S_{\mathrm{cl}}\), but it does not change the certified lower bound obtained from any imperfect observed score \(Q<1\), for which the inversion remains \(\eta \ge \max\{0,Q-S_{\mathrm{cl}}\}\).

Fix any admissible benchmark-dependent model. Define
\begin{equation}
f_z(\lambda):=\sum_{a,b} p(a|x(z),\lambda)\,p(b|y(z),\lambda)\,r(z,a,b).
\label{eq:fz}
\end{equation}
Since \(0\le r\le 1\), one has \(0\le f_z(\lambda)\le 1\). Hence
\begin{equation}
S=\sum_z q(z)\sum_\lambda p(\lambda|z)\,f_z(\lambda).
\end{equation}
Add and subtract \(p_0(\lambda)\):
\begin{equation}
\begin{aligned}
S &=
\sum_z q(z)\sum_\lambda p_0(\lambda)\,f_z(\lambda)
\\\qquad &+
\sum_z q(z)\sum_\lambda [p(\lambda|z)-p_0(\lambda)]f_z(\lambda).
\label{eq:addsub}
\end{aligned}
\end{equation}
The first term is at most \(S_{\mathrm{cl}}\) by \eqref{eq:Scl}. For the second term, the variational bound for \(0\le f\le 1\) gives
\begin{equation}
\left|
\sum_\lambda (P(\lambda)-Q(\lambda))f(\lambda)
\right|
\le \TV(P,Q).
\label{eq:varbound}
\end{equation}
Applying \eqref{eq:varbound} with \(P(\lambda)=p(\lambda|z)\), \(Q(\lambda)=p_0(\lambda)\), and \(f=f_z\), we obtain
\begin{equation}
\sum_\lambda [p(\lambda|z)-p_0(\lambda)]f_z(\lambda)\le \eta.
\end{equation}
Averaging over \(z\) yields \(S\le S_{\mathrm{cl}}+\eta\). Taking the supremum over all admissible models proves \eqref{eq:general_bound}.
\end{proof}

For $n-$repeated tasks with average score
\begin{equation}
\bar S_n:=\frac1n\sum_{t=1}^n \E[r_t],
\end{equation}
the general relaxed advantage inequality immediately gives
\begin{equation}
\bar S_{n,\eta}\le \bar S_{\mathrm{cl},n}+\eta.
\label{eq:average_bound}
\end{equation}
If roundwise dependence parameters \(\eta_t\) are used, then
\begin{equation}
\bar S_{n,\eta_1,\dots,\eta_n}
\le
\bar S_{\mathrm{cl},n}+\frac1n\sum_{t=1}^n \eta_t.
\label{eq:average_roundwise}
\end{equation}

\paragraph*{Exact saturation:} The relaxed inequality~\eqref{eq:general_bound} is not only valid but, in general, tight: there exist tasks for which
benchmark dependence increases the achievable classical score by exactly \(\eta\).

\begin{proof}
Let \(m\) denote the number of possible hidden task instances with $m\ge2$. Consider a task where the role of Bob is trivial, \(\mathcal B=\{0\}\), which is expressed as \(\mathcal Z=[m]:=\{1,\dots,m\},\qquad q(z)=\frac1m,\qquad \mathcal A=[m],\qquad \mathcal B=\{0\},
\) and reward
\begin{equation}
r(z,a,b)=\1\{a=z\}.
\label{eq:hidden_reward}
\end{equation}

Then
\begin{equation}
S_{\mathrm{cl}}=\frac1m,
\qquad
S_\eta=\min\!\left\{1,\frac1m+\eta\right\}.
\label{eq:hidden_exact}
\end{equation}

In the loophole-free model, the output law \(p(a)\) is independent of \(z\), so
\begin{equation}
S=\frac1m\sum_{z=1}^m p(a=z)=\frac1m\sum_{a=1}^m p(a)=\frac1m.
\end{equation}
Thus \(S_{\mathrm{cl}}=1/m\). Using  equation ~\eqref{eq:general_bound},
\begin{equation}
S_\eta\le \min\!\left\{1,\frac1m+\eta\right\}.
\label{eq:hidden_upper}
\end{equation}

Equation~\eqref{eq:hidden_upper} gives the required upper bound on \(S_\eta\). To prove
equality in \eqref{eq:hidden_exact}, we now construct explicit admissible models that achieve
this bound.

For \(0\le \eta \le 1-\frac1m\), let the reference hidden-variable distribution be uniform, $p_0(\lambda)=\frac1m, \lambda\in[m],$
and let Alice output the hidden variable deterministically $a=\lambda$. For each \(z\in[m]\), define the benchmark-dependent hidden-variable law by
\begin{equation}
p(\lambda|z)=
\begin{cases}
\frac1m+\eta, & \lambda=z,\\[1ex]
\frac1m-\dfrac{\eta}{m-1}, & \lambda\neq z.
\end{cases}
\label{eq:hidden_constr1}
\end{equation}
This construction shifts probability mass \(\eta\) toward the correct label \(\lambda=z\) while
redistributing the same total amount uniformly over the remaining \(m-1\) values, so that
\(p(\lambda|z)\) remains normalized. The condition \(0\le \eta \le 1-\frac1m\) ensures that all
probabilities remain nonnegative.

Its total-variation distance from the uniform reference is
\begin{equation}
\begin{split}
\TV\!\bigl(p(\lambda|z),p_0(\lambda)\bigr)
&=
\frac12
\left|\left(\frac1m+\eta\right)-\frac1m\right|
+\\& \frac12
(m-1)\left|
\left(\frac1m-\frac{\eta}{m-1}\right)-\frac1m
\right|
\nonumber\\
&=
\frac12\left(\eta+(m-1)\frac{\eta}{m-1}\right)
=\eta.
\label{eq:hidden_tv}
\end{split}
\end{equation}
Since the output rule is \(a=\lambda\), the strategy succeeds exactly when \(\lambda=z\). Hence
\[
p(a=z|z)=p(\lambda=z|z)=\frac1m+\eta,
\]
and therefore the average score is
\[
S=\frac1m+\eta.
\]

For \(\eta\ge 1-\frac1m\), choose again \(p_0(\lambda)=1/m\), \(a=\lambda\), and
\begin{equation}
p(\lambda|z)=\1\{\lambda=z\}.
\label{eq:hidden_constr2}
\end{equation}
Then \(\TV(p(\lambda|z),p_0)=1-\frac1m\le \eta\), and \(S=1\). Combining both constructions with \eqref{eq:hidden_upper} proves \eqref{eq:hidden_exact}.
\end{proof}

\paragraph*{Universal optimality:}
\label{cor:nosublinear} No universal refinement can reduce the linear \(\eta\) penalty: for some tasks, benchmark dependence translates directly into an equal increase in achievable score. In other words, there is no universal bound of the form
\begin{equation}
S_\eta \le S_{\mathrm{cl}}+g(\eta)
\end{equation}
valid for all bounded-reward tasks with \(g(\eta)=o(\eta)\) as \(\eta\to 0\). The same hidden-instance task, repeated independently and scored by an average of rounds, saturates the additive average-score bound \eqref{eq:average_roundwise}. Thus the additive dependence in \eqref{eq:average_roundwise} is also generically optimal.

\section{Results}
\subsection{Repeated Bell-derived tasks}

The universal bound~\eqref{eq:general_bound} is intentionally model-agnostic. For structured repeated tasks, however, that generality can be traded for sharper certification: the task architecture itself constrains how benchmark dependence can enhance classical performance, leading to explicit loophole thresholds for reproducing a target score.

Let \(G\) be a two-party primitive with instance set
\(\mathcal Z_1\), distribution \(q_1\), reward
\(r_1(z,a,b)\in[0,1]\), and one-round loophole-free classical value
\(\omega_{\mathrm c}\), equal to \(S_{\mathrm{cl}}\) for this primitive:
\begin{equation}
\begin{aligned}
\omega_{\mathrm c}
:=
\sup_{p_0,\,p(a|x,\lambda),\,p(b|y,\lambda)}
\sum_{z\in\mathcal Z_1}q_1(z)
\sum_{\lambda}p_0(\lambda)
\sum_{a,b}\,
&p(a|x(z),\lambda)
\\[-1mm]
{}\times&
p(b|y(z),\lambda)\,
r_1(z,a,b).
\end{aligned}
\label{eq:omega_c}
\end{equation}
Define the \(n\)-round product task by
\begin{equation}
\begin{aligned}
q_n(z_1,\dots,z_n):=\prod_{i=1}^n q_1(z_i),
\\
r_n(z,a,b):=\prod_{i=1}^n r_1(z_i,a_i,b_i).
\label{eq:product_task}
\end{aligned}
\end{equation}

Product scoring is not the default benchmark in noisy Bell experiments, where one usually tracks average win rates. It is, however, the natural figure of merit when repeated instances are bundled into a single block that is accepted only if every constituent round succeeds. In that exact-certification regime, one defines the block reward by \(r_n=\prod_{i=1}^n r_1^{(i)}\), so the score is precisely the probability of flawless completion of the full repeated task.

We consider the roundwise benchmark-dependent classical model
\begin{equation}
p(a,b|z)
=
\sum_{\lambda_1,\dots,\lambda_n}
\prod_{i=1}^n p(\lambda_i|z_i)\,p(a_i|x(z_i),\lambda_i)\,p(b_i|y(z_i),\lambda_i),
\label{eq:roundwise_model}
\end{equation}
with roundwise benchmark-dependence parameters
\begin{equation}
\eta_i:=\sup_{z_i\in\mathcal Z_1}\TV\!\bigl(p(\lambda_i|z_i),p_{0,i}(\lambda_i)\bigr).
\label{eq:eta_i}
\end{equation}

\paragraph*{Task-specific sharpening for roundwise benchmark dependence:}
\label{thm:product}
For the \(n\)-round product task \eqref{eq:product_task} under the roundwise model \eqref{eq:roundwise_model}, we posit that 
\begin{equation}
S_{\eta_1,\dots,\eta_n}^{(n)}
\le
\prod_{i=1}^n (\omega_{\mathrm c}+\eta_i).
\label{eq:product_bound}
\end{equation}
In particular, if \(\eta_i=\eta\) for all \(i\), then
\begin{equation}
S_{\eta}^{(n)}\le (\omega_{\mathrm c}+\eta)^n.
\label{eq:product_bound_uniform}
\end{equation}
Operationally,~\eqref{eq:product_bound} converts a measured 
$n$-round score into a direct lower bound on the benchmark dependence that any classical surrogate would need in order to reproduce it.
\begin{proof}
Using \eqref{eq:roundwise_model} and \eqref{eq:product_task},
\begin{equation}
\begin{aligned}
S^{(n)}
&=
\sum_{\mathbf z,\mathbf a,\mathbf b,\boldsymbol\lambda}
\prod_{i=1}^n
q_1(z_i)\,
p(\lambda_i|z_i)\,\\&
p(a_i|x(z_i),\lambda_i)\,
p(b_i|y(z_i),\lambda_i)\,
r_1(z_i,a_i,b_i)
\nonumber\\
&=
\prod_{i=1}^n
\sum_{z_i,a_i,b_i,\lambda_i}
q_1(z_i)\,
p(\lambda_i|z_i)\,\\&
p(a_i|x(z_i),\lambda_i)\,
p(b_i|y(z_i),\lambda_i)\,
r_1(z_i,a_i,b_i)
,
\label{eq:factor_score}
\end{aligned}
\end{equation}
where \(\mathbf z=(z_1,\dots,z_n)\), \(\mathbf a=(a_1,\dots,a_n)\), \(\mathbf b=(b_1,\dots,b_n)\), and \(\boldsymbol\lambda=(\lambda_1,\dots,\lambda_n)\).
Each bracket is a one-round benchmark-dependent classical score. By ~\eqref{eq:general_bound}, it is at most \(\omega_{\mathrm c}+\eta_i\). Therefore
\begin{equation}
S^{(n)}\le \prod_{i=1}^n (\omega_{\mathrm c}+\eta_i),
\end{equation}
which proves \eqref{eq:product_bound}. Equation \eqref{eq:product_bound_uniform} is the special case \(\eta_i=\eta\).
\end{proof}

\paragraph*{Exact saturation of the multiplicative bound:}
\label{thm:product_tight}
Under the assumptions of~\eqref{eq:product_bound}, the multiplicative bound is exact in general. Specifically, for the hidden-instance identification primitive with one-round value \(S_{\mathrm{cl}}=\omega_{\mathrm c}=1/m\), the \(n\)-round product task satisfies
\begin{equation}
S_{\eta}^{(n)}=(\omega_{\mathrm c}+\eta)^n
\qquad
\text{for } 0\le \eta \le 1-\omega_{\mathrm c}.
\label{eq:product_exact}
\end{equation}

\begin{proof}
Use the one-round saturation construction of~\eqref{eq:hidden_exact} independently in each round. Since the model and reward factorize across rounds, the total score factorizes as in \eqref{eq:factor_score}, and each one-round factor equals \(\omega_{\mathrm c}+\eta\). Thus \(S_\eta^{(n)}=(\omega_{\mathrm c}+\eta)^n\).
\end{proof}

Inter-round hidden-state correlations can invalidate the factorization
in Eq.~\eqref{eq:factor_score}; therefore,
Eqs.~\eqref{eq:product_bound} and \eqref{eq:product_exact} are specific
to roundwise benchmark dependence. In contrast, the additive average-score bound \eqref{eq:average_bound} remains valid without any factorization assumption.

\paragraph*{Pathwise benchmark dependence under correlated rounds:}
\label{prop:pathwise}
The pathwise formulation trades sharpness for generality: by absorbing the entire hidden trajectory into a single augmented hidden variable, it yields a loophole-robust bound that remains valid even when the hidden dynamics are arbitrarily correlated across rounds. Let \(\Lambda_{1:n}\) be the full hidden-state trajectory and define the pathwise dependence parameter
\begin{equation}
\eta_{\mathrm{path}}
:=
\sup_{z_{1:n}}
\TV\!\bigl(p(\Lambda_{1:n}|z_{1:n}),p_0(\Lambda_{1:n})\bigr),
\label{eq:eta_path}
\end{equation}
where \(p_0(\Lambda_{1:n})=\sum_{z_{1:n}}q_n(z_{1:n})p(\Lambda_{1:n}|z_{1:n})\). Then for any \(n\)-round bounded-reward task with any certified loophole-free classical bound \(C_n\),
\begin{equation}
S^{(n)}\le C_n+\eta_{\mathrm{path}}.
\label{eq:pathwise_bound}
\end{equation}
Equation~\eqref{eq:pathwise_bound} is simply the general relaxed inequality applied to the augmented hidden variable \(\Lambda_{1:n}\) and the full \(n\)-round task regarded as a single bounded-reward instance.

This bound is especially useful when the hidden process is not roundwise
factorized. As a concrete example, consider a benchmark-dependent Markov
model of the form
\begin{equation}
p(\Lambda_{1:n}|z_{1:n})
=
\mu_{z_1}(\lambda_1)
\prod_{t=2}^{n}
K_t^{z_t}(\lambda_{t-1},\lambda_t),
\label{eq:benchmark_markov_model}
\end{equation}
where \(\mu_{z_1}\) is the initial distribution of the hidden state
\(\lambda_1\), conditioned on the first benchmark instance \(z_1\), and
\(K_t^{z_t}(\lambda_{t-1},\lambda_t)\) is the transition kernel from
\(\lambda_{t-1}\) to \(\lambda_t\) at round \(t\), allowed to depend on
the contemporaneous benchmark instance \(z_t\). Define the one-step
benchmark sensitivities
\begin{align}
\varepsilon_1
&:=
\sup_{u,v}
\TV(\mu_u,\mu_v),
\label{eq:epsilon_initial}
\\
\varepsilon_t
&:=
\sup_{u,v,\lambda}
\TV\!\left(
K_t^{u}(\lambda,\cdot),
K_t^{v}(\lambda,\cdot)
\right),
\qquad t\geq 2.
\label{eq:epsilon_kernel}
\end{align}

\begin{proposition}[Sequential path-coupling bound]
\label{prop:sequential_path_coupling}
For the benchmark-dependent Markov model
\eqref{eq:benchmark_markov_model}, the pathwise dependence satisfies
\begin{equation}
\eta_{\mathrm{path}}
\leq
1-(1-\varepsilon_1)
\prod_{t=2}^{n}(1-\varepsilon_t).
\label{eq:pathwise_coupling_bound}
\end{equation}
Consequently,
\begin{equation}
\eta_{\mathrm{path}}
\leq
\min\!\left\{
1,\,
\varepsilon_1+\sum_{t=2}^{n}\varepsilon_t
\right\}.
\label{eq:pathwise_union_bound}
\end{equation}
\end{proposition}

\begin{proof}
For each benchmark sequence \(z_{1:n}\), let
\[
P_{z_{1:n}}
:=
p(\Lambda_{1:n}|z_{1:n})
\]
denote the corresponding path law. Consider two arbitrary benchmark
sequences \(u_{1:n}\) and \(v_{1:n}\). We construct a coupling
\((\Lambda_{1:n},\Lambda'_{1:n})\) of
\(P_{u_{1:n}}\) and \(P_{v_{1:n}}\) sequentially.

At the first round, maximally couple
\(\mu_{u_1}\) and \(\mu_{v_1}\). By the definition of
\(\varepsilon_1\),
\begin{equation}
\Prob(\Lambda_1=\Lambda'_1)
\geq
1-\varepsilon_1.
\label{eq:first_round_coupling}
\end{equation}
Suppose that the two trajectories agree through round \(t-1\), with
common state
\(\Lambda_{t-1}=\Lambda'_{t-1}=\lambda\). At round \(t\), maximally
couple
\[
K_t^{u_t}(\lambda,\cdot)
\quad\text{and}\quad
K_t^{v_t}(\lambda,\cdot).
\]
Equation~\eqref{eq:epsilon_kernel} then gives
\begin{equation}
\Prob\!\left(
\Lambda_t=\Lambda'_t
\,\middle|\,
\Lambda_{1:t-1}=\Lambda'_{1:t-1}
\right)
\geq
1-\varepsilon_t.
\label{eq:conditional_step_coupling}
\end{equation}
Multiplying these conditional agreement probabilities yields
\begin{equation}
\Prob(\Lambda_{1:n}=\Lambda'_{1:n})
\geq
(1-\varepsilon_1)
\prod_{t=2}^{n}(1-\varepsilon_t).
\label{eq:path_agreement_probability}
\end{equation}
The coupling characterization of total variation therefore implies
\begin{equation}
\TV\!\left(
P_{u_{1:n}},
P_{v_{1:n}}
\right)
\leq
1-(1-\varepsilon_1)
\prod_{t=2}^{n}(1-\varepsilon_t).
\label{eq:pairwise_path_tv}
\end{equation}

The benchmark marginal is the mixture
\[
p_0(\Lambda_{1:n})
=
\sum_{v_{1:n}}
q_n(v_{1:n})P_{v_{1:n}}.
\]
By convexity of total variation in its second argument,
\begin{align}
\TV\!\left(
P_{u_{1:n}},
p_0
\right)
&\leq
\sum_{v_{1:n}}
q_n(v_{1:n})
\TV\!\left(
P_{u_{1:n}},
P_{v_{1:n}}
\right)
\nonumber\\
&\leq
1-(1-\varepsilon_1)
\prod_{t=2}^{n}(1-\varepsilon_t).
\end{align}
Taking the supremum over \(u_{1:n}\) proves
Eq.~\eqref{eq:pathwise_coupling_bound}. Finally,
\[
1-\prod_{t=1}^{n}(1-\varepsilon_t)
\leq
\sum_{t=1}^{n}\varepsilon_t,
\]
which, together with \(\eta_{\mathrm{path}}\leq 1\), proves
Eq.~\eqref{eq:pathwise_union_bound}.
\end{proof}

A useful sufficient condition can be expressed through common
minorization measures. Suppose that the initial distributions and
transition kernels satisfy
\begin{align}
\mu_z(\cdot)
&\geq
\alpha_1\nu_1(\cdot)
\qquad
\text{for all }z,
\label{eq:initial_minorization}
\\
K_t^z(\lambda,\cdot)
&\geq
\alpha_t\nu_t(\cdot)
\qquad
\text{for all }z,\lambda,\quad t\geq 2,
\label{eq:kernel_minorization}
\end{align}
where \(0\leq\alpha_t\leq 1\) and each \(\nu_t\) is a probability
distribution. Any two distributions sharing a common component of
weight \(\alpha_t\) have total-variation distance at most
\(1-\alpha_t\). Hence
\[
\varepsilon_t\leq 1-\alpha_t
\qquad
(t=1,\ldots,n),
\]
and Proposition~\ref{prop:sequential_path_coupling} gives
\begin{equation}
\eta_{\mathrm{path}}
\leq
1-\prod_{t=1}^{n}\alpha_t.
\label{eq:pathwise_minorization_bound}
\end{equation}

These estimates control pathwise benchmark dependence using local
deviations of the initial distribution and transition kernels. The
result is intentionally conservative because it concerns the law of the
entire trajectory: once two paths disagree at an early round, that
discrepancy remains recorded in \(\Lambda_{1:n}\). Dobrushin-type
contraction coefficients can sharpen corresponding bounds for the
current-state or suffix distributions, but they do not erase
discrepancies already contained in the full path
\cite{DoucMoulinesRosenthal2004,Mitrophanov2005,
HairerMattingly2011,RudolfSchweizer2018}. Thus,
Eq.~\eqref{eq:pathwise_bound} complements the sharper roundwise and
average-score bounds when inter-round hidden-state correlations cannot
be neglected.
\subsection{Finite-sample certification}

Suppose a task is implemented in \(N\) independent trials (or shots) under fixed conditions producing i.i.d.\ bounded scores \(W_1,\dots,W_N\in[0,1]\), with empirical mean
\begin{equation}
\widehat S_N:=\frac1N\sum_{j=1}^N W_j.
\label{eq:Shat}
\end{equation}
By Hoeffding's inequality, with probability at least \(1-\alpha\), the true mean score \(S\) satisfies
\begin{equation}
S \ge L_N(\alpha)
:=
\widehat S_N-\sqrt{\frac{\log(1/\alpha)}{2N}},
\label{eq:hoeffding_lcb}
\end{equation}
where \(L_N(\alpha)\) is the corresponding one-sided Hoeffding lower confidence bound.

Therefore, if the relevant classical benchmark is \(C\), then for any candidate benchmark-dependence level \(\eta_0\), a one-sided level-\(\alpha\) test of the null hypothesis
\begin{equation}
H_0(\eta_0):\quad S \le C+\eta_0
\label{eq:null_eta0}
\end{equation}
rejects whenever
\begin{equation}
L_N(\alpha)>C+\eta_0.
\label{eq:test_rule}
\end{equation}
Equivalently, a certified lower bound on the loophole strength at confidence \(1-\alpha\) is
\begin{equation}
\eta_{\min}^{\mathrm{cert}}(\alpha)
:=
\max\!\left\{0,\;L_N(\alpha)-C\right\}.
\label{eq:eta_cert}
\end{equation}
If \(W_j\in\{0,1\}\), one may replace the Hoeffding bound by the exact one-sided Clopper--Pearson lower confidence limit \cite{ClopperPearson1934, BrownCaiDasGupta2001}. Writing \(K=\sum_{j=1}^N W_j\), the corresponding \((1-\alpha)\) lower bound on the true success probability \(S\) is
\begin{equation}
L_N(\alpha)=
\begin{cases}
0, & K=0,\\[4pt]
\mathrm{B}^{-1}\!\bigl(\alpha;\,K,\,N-K+1\bigr), & K\ge 1,
\end{cases}
\label{eq:cp_lower}
\end{equation}
where \(\mathrm{B}^{-1}(\alpha;a,b)\) is the \(\alpha\)-quantile of the \(\mathrm{Beta}(a,b)\) distribution.
For product-task scores obeying \eqref{eq:product_bound_uniform}, a lower confidence bound \(L_N(\alpha)\) on the observed product score yields
\begin{equation}
\eta_{\min}^{\mathrm{cert}}(\alpha)
=
\max\!\left\{0,\;L_N(\alpha)^{1/n}-\omega_{\mathrm c}\right\}.
\label{eq:eta_cert_product}
\end{equation}

\subsection{Mutual-information dependence}

A worst-case loophole parameter is the right notion for adversarial certification, but it can be unnecessarily conservative when benchmark dependence is concentrated on rare instances. To complement the worst-case parameter \(\eta_{\mathrm{bm}}\), we therefore introduce an average-case information-theoretic measure that quantifies how much information about the benchmark instance \(Z\) is encoded, on average, in the hidden variable \(\Lambda\)~\cite{CoverThomas2006}.

We define the mutual information ($\I(\Lambda;Z)$), which is a measure of the average benchmark information carried by the hidden variable, between \(\Lambda\) and \(Z\) by
\begin{equation}
\I(\Lambda;Z)
=
\sum_z q(z)\,\KL\!\bigl(p(\cdot|z)\,\|\,p_0\bigr),
\label{eq:MI}
\end{equation}
where $\KL(P\|Q)$ is the Kullback--Leibler divergence given by $\sum_{\lambda} P(\lambda)\log\frac{P(\lambda)}{Q(\lambda)}.$ 

To compare with the total-variation formulation, define the average total-variation distance
\begin{equation}
\bar\eta
:=
\sum_z q(z)\,\TV\!\bigl(p(\lambda|z),p_0(\lambda)\bigr).
\label{eq:etabar}
\end{equation}
Applying Pinsker's inequality pointwise,
\begin{equation}
\TV(P,Q)\le \sqrt{\frac12\,\KL(P\|Q)},
\label{eq:pinsker_pointwise}
\end{equation}
and then Jensen's inequality for the concave square-root function, we obtain
\begin{equation}
\bar\eta
\le
\sum_z q(z)\sqrt{\frac12\,\KL\!\bigl(p(\cdot|z)\,\|\,p_0\bigr)}
\le
\sqrt{\frac12\,\I(\Lambda;Z)}.
\label{eq:pinsker}
\end{equation}

Using the general relaxed inequality~\eqref{eq:general_bound}, but averaging over \(z\) rather than taking the supremum over \(z\), yields
\begin{equation}
S\le S_{\mathrm{cl}}+\bar\eta
\le
S_{\mathrm{cl}}+\sqrt{\frac12\,\I(\Lambda;Z)}.
\label{eq:MI_bound}
\end{equation}
Equation~\eqref{eq:MI_bound} is thus an average-case relaxed advantage inequality: the excess score above the loophole-free classical benchmark is controlled by the average information that the hidden variable carries about the benchmark.

This formulation is genuinely weaker than the worst-case bound. The reason is structural: \(\I(\Lambda;Z)\) averages over benchmark instances, whereas \(\eta_{\mathrm{bm}}\) is a supremum over them. Consequently, a rare benchmark instance can exhibit \(O(1)\) total-variation deviation while contributing arbitrarily little to \(\I(\Lambda;Z)\), so \(\eta_{\mathrm{bm}}\) cannot, in general, be bounded in terms of \(\I(\Lambda;Z)\) alone.

\subsection{Standard Bell-derived thresholds}

Here, we convert the abstract benchmark-dependence parameter into an experimentally interpretable loophole threshold for standard Bell-derived tasks. For a given one-round nonlocal-game primitive \cite{Almeida2010}, let \(\omega_{\mathrm c}\) denote the optimal loophole-free classical score, let \(\omega_{\mathrm q}\) denote the optimal entangled (quantum) score, let \(Q_n\) denote the observed score of the \(n\)-round product task, which is the probability of winning all \(n\) rounds simultaneously under the product reward, and let \(\bar Q\) denote the observed average one-round score. Under the roundwise benchmark-dependent model with uniform parameter \(\eta\),~\eqref{eq:product_bound} implies
\begin{equation}
Q_n \le (\omega_{\mathrm c}+\eta)^n.
\label{eq:threshold_start}
\end{equation}
Equivalently, any classical surrogate consistent with the observed product score \(Q_n\) must satisfy
\begin{equation}
\eta \ge \max\!\left\{0,\,Q_n^{1/n}-\omega_{\mathrm c}\right\}.
\label{eq:imperfect_threshold_general}
\end{equation}
Likewise, for the average one-round score, the additive bound gives
\begin{equation}
\eta \ge \max\!\left\{0,\,\bar Q-\omega_{\mathrm c}\right\}.
\label{eq:imperfect_threshold_average}
\end{equation}

This inversion is especially transparent for pseudo-telepathy primitives \cite{BrassardBroadbentTapp2005, HomaBodorBernad2026, Pawela2013, KelleherRoomyHolweck2024}, namely Bell-derived games for which entangled players can win with certainty, \(\omega_{\mathrm q}=1\), while every classical strategy satisfies \(\omega_{\mathrm c}<1\) \cite{Brassard2005,Brunner2014}. In that case, an ideal perfect product-task score \(Q_n=1\) forces
\begin{equation}
\eta \ge 1-\omega_{\mathrm c}.
\label{eq:perfect_threshold_general}
\end{equation}
Thus the threshold is simply the gap between perfect task performance and the best loophole-free classical score.

Among the standard pseudo-telepathy examples (see Table ~\ref{tab:thresholds}), the three-party GHZ/Mermin game has \(\omega_{\mathrm c}=3/4\) and \(\omega_{\mathrm q}=1\) \cite{Greenberger1990,Mermin1990AJP,BrassardBroadbentTapp2005}, while the Mermin--Peres magic-square game has \(\omega_{\mathrm c}=8/9\) and \(\omega_{\mathrm q}=1\) \cite{Mermin1990PRL,Peres1990,Brassard2005}. The CHSH primitive should be treated separately because its optimal entangled value is not unity; rather, \(\omega_{\mathrm q}=\cos^2(\pi/8)\) \cite{Clauser1969,Cleve2004,Brunner2014}. Hence the relevant ideal-score threshold for CHSH follows from Eq.~\eqref{eq:imperfect_threshold_general}, not Eq.~\eqref{eq:perfect_threshold_general}, and is
\begin{equation}
\eta \ge \cos^2\!\left(\frac{\pi}{8}\right)-\frac34 \approx 0.1036.
\label{eq:chsh_threshold}
\end{equation}
This distinction matters operationally: pseudo-telepathy primitives quantify how much benchmark dependence is required to fake \emph{perfect} nonclassical task performance, whereas CHSH quantifies the loophole strength needed to fake the smaller, but still strictly nonclassical, Tsirelson-limited score.

\begin{table}[h]
\caption{Minimum roundwise benchmark dependence required for a
classical surrogate to reproduce the optimal quantum score of each
primitive under the product-task model. For GHZ/Mermin and magic
square, the optimal entangled value is
\(\omega_{\mathrm q}=1\), whereas for CHSH it is
\(\omega_{\mathrm q}=\cos^2(\pi/8)\).}
\label{tab:thresholds}
\begin{ruledtabular}
\begin{tabular}{lcc}
Primitive & \(\omega_{\mathrm c}\) & \(\eta_{\min}^{\mathrm{ideal}}\) \\
\hline
CHSH & \(3/4\) & \(0.1036\) \\
GHZ/Mermin & \(3/4\) & \(1/4\) \\
Magic square & \(8/9\) & \(1/9\) \\
\end{tabular}
\end{ruledtabular}
\end{table}

In particular, for the Mermin--Peres magic-square game, any roundwise benchmark-dependent classical explanation of a perfect repeated product-task score must satisfy
\begin{equation}
\eta \ge \frac19.
\label{eq:magic_square_threshold}
\end{equation}

\subsection{Multi-party extension}

The two-party results are not special to bipartite tasks. The same loophole-robust logic extends directly to distributed tasks with an arbitrary number of parties: increasing the number of agents changes the task-specific classical benchmark, but not the universal way in which benchmark dependence enters the bound. This is important because multipartite Bell scenarios support stronger forms of nonclassicality, including genuine multipartite nonlocality, yet the certification question remains the same: how much benchmark dependence would a classical surrogate need in order to reproduce the observed score?

Consider a \(k\)-party distributed task. In each round, a benchmark instance \(z\) is drawn from a known distribution \(q(z)\), and it specifies local inputs \(x_1(z),\dots,x_k(z)\) delivered to the \(k\) parties. The parties return outputs \(a_1,\dots,a_k\), and their performance is scored by a bounded reward function
\begin{equation}
0 \le r(z,a_1,\dots,a_k) \le 1.
\label{eq:multipartite_reward}
\end{equation}
A benchmark-dependent classical model is then
\begin{equation}
p(a_1,\dots,a_k|z)
=
\sum_\lambda p(\lambda|z)\prod_{j=1}^k p(a_j|x_j(z),\lambda),
\label{eq:multipartite_model}
\end{equation}
where \(\lambda\) is a shared hidden variable and the conditional response \(p(a_j|x_j(z),\lambda)\) of each party depends only on its own local input and on \(\lambda\). Thus, the only departure from the loophole-free model is that the hidden-variable distribution is allowed to depend on the benchmark instance through \(p(\lambda|z)\), exactly as in the bipartite case.

With this factorized \(k\)-party classical structure in place, the proof of~\eqref{eq:general_bound} is unchanged: one again isolates the benchmark-independent reference term and bounds the remainder by total variation. Therefore,
\begin{equation}
S_\eta \le S_{\mathrm{cl}}+\eta,
\label{eq:multipartite_bound}
\end{equation}
where \(S_{\mathrm{cl}}\) is now the optimal loophole-free classical value for the \(k\)-party task. The significance of Eq.~\eqref{eq:multipartite_bound} is that the additive penalty remains exactly \(\eta\), independent of \(k\). In other words, adding more parties can make the task itself more nonclassical by lowering the classical benchmark, but it does not amplify the universal worst-case loophole correction.

The same conclusion holds for repeated product tasks, provided one retains the roundwise benchmark-dependent structure used earlier. If the \(n\)-round task is built from independent repetitions of a one-round \(k\)-party primitive and each round obeys the same dependence bound \(\eta\), then the multiplicative argument also carries over:
\begin{equation}
S_\eta^{(n)} \le (\omega_{\mathrm c}+\eta)^n,
\label{eq:multipartite_product}
\end{equation}. Hence the number of parties enters only through the task-dependent benchmark \(\omega_{\mathrm c}\), not through any change in the universal coefficient multiplying \(\eta\). This clean separation is useful in multipartite settings, where different tasks probe different strengths of nonlocality, from GHZ-type contradictions to Svetlichny-type genuine multipartite nonlocality, but the loophole accounting retains the same operational form.

\subsection{Illustrative certification example: noisy magic-square data.}
To illustrate the finite-sample certification procedure, consider a synthetic noisy realization of the Mermin--Peres magic-square primitive, whose loophole-free classical value is $\omega_{\mathrm c}=\frac89,$
while the ideal entangled value is \(1\)~\cite{Mermin1990PRL,Peres1990}. Since the round outcomes are binary, \(W_j\in\{0,1\}\), the appropriate exact one-sided finite-sample bound is the Clopper--Pearson lower confidence limit~\cite{ClopperPearson1934}.

\emph{Average-score certification:}
Suppose \(N=10^4\) independent rounds yield \(K=9650\) wins, so that the observed average score is $\widehat{\bar Q}=\frac{K}{N}=0.965.$ At confidence level \(1-\alpha=0.95\), the exact one-sided Clopper--Pearson lower bound is $L_N^{\mathrm{CP}}(0.05)\approx 0.9618.$ Applying Eq.~\eqref{eq:imperfect_threshold_average} gives the certified lower bound
\[
\eta_{\min}^{\mathrm{cert}}(0.05)
=
L_N^{\mathrm{CP}}(0.05)-\omega_{\mathrm c}
\approx 0.9618-\frac89
\approx 0.0729.
\]
Thus, with \(95\%\) confidence, any classical surrogate reproducing the observed average score must have benchmark dependence at least \(7.29\times10^{-2}\).

\emph{Product-score certification:}
Now group the data into \(N_{\mathrm{blk}}=2000\) independent blocks of \(n=3\) rounds each, and assign block reward \(1\) only if all three rounds are won. Suppose \(1800\) of the \(2000\) blocks are flawless, giving the observed product score $\widehat Q_3=\frac{1800}{2000}=0.900.$ The exact one-sided \(95\%\) Clopper--Pearson lower bound for the block-success probability is $L_{N_{\mathrm{blk}}}^{\mathrm{CP}}(0.05)\approx 0.8883.$
Equation~\eqref{eq:imperfect_threshold_general} then yields
\[
\eta_{\min}^{\mathrm{cert}}(0.05)
=
\bigl(L_{N_{\mathrm{blk}}}^{\mathrm{CP}}(0.05)\bigr)^{1/3}-\omega_{\mathrm c}
\approx
0.8883^{1/3}-\frac89
\approx 0.0724.
\]
Thus, with \(95\%\) confidence, any classical explanation of the observed three-round flawless-block rate requires benchmark dependence at least \(7.24\times10^{-2}\).

These two calculations use the same certification logic but emphasize different operational notions of success. Average-score certification is statistically more forgiving and is the natural benchmark for noisy experiments, whereas product-score certification directly quantifies the benchmark dependence required to explain exact blockwise success. In either case, the result is a one-sided confidence set $\eta\in[\eta_{\min}^{\mathrm{cert}}(\alpha),1]$
at confidence level \(1-\alpha\).

The certified value of \(\eta\) has a direct operational meaning: for example, \(\eta_{\min}^{\mathrm{cert}}(0.05)\approx 0.0724\) for \(n=3\) means that, with \(95\%\) confidence, any classical surrogate reproducing the observed three-round flawless-block rate must exhibit at least \(7.24\%\) worst-case benchmark dependence per round. In Bell language, where the benchmark instance is the full setting tuple, this implies a Hall-type measurement-dependence parameter \(M\ge 2\eta_{\min}^{\mathrm{cert}}\approx 0.145\), or equivalently a remaining measurement-independence fraction \(F\le 1-\eta_{\min}^{\mathrm{cert}}\approx 0.928\)~\cite{Hall2010,Hall2011,FriedmanEtAl2019}.

\subsection{Experimental finite-sample certification from IBM backends}
\label{subsec:ibm_block_certification}

\begin{figure*}[t]
    \centering
    \includegraphics[width=0.98\textwidth]{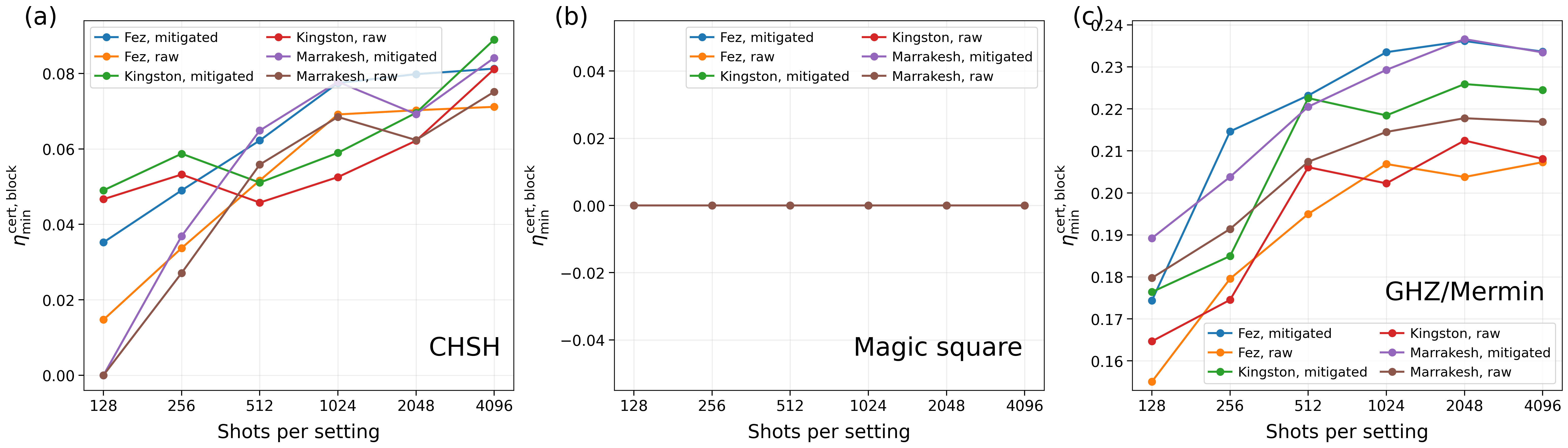}
    \caption{
\textbf{Reconstructed block analysis for Bell-derived benchmark
primitives on IBM quantum backends.}
The exact raw-count benchmark-dependence certificate
and the readout-mitigated plug-in
estimate are shown as functions
of shots per benchmark context for \textbf{(a)} CHSH,
\textbf{(b)} the Mermin--Peres magic-square game, and
\textbf{(c)} the Mermin--GHZ game. Raw-count certificates use
simultaneous one-sided Clopper--Pearson lower limits with Bonferroni
allocation over contexts. Mitigated values are obtained from full
assignment-matrix correction and are displayed as sensitivity estimates,
not as exact-confidence bounds, because the corrected frequencies are
not binomial counts. CHSH gives moderate positive raw certification,
magic square remains uncertified under the nine-context product
reconstruction, and Mermin--GHZ gives the strongest raw and mitigated
performance.
}
    \label{fig:finite_sample_block_certification}
\end{figure*}

We next apply the finite-sample certification framework to data collected
on IBM Fez, IBM Kingston, and IBM Marrakesh for CHSH, the
Mermin--Peres magic-square game, and the Mermin--GHZ game. The
experiments were performed over multiple shot budgets, and the raw and
readout-mitigated results were analyzed separately. Because the exported
data were aggregated by benchmark context rather than retained as a
shot-ordered sequence of rounds, the product-task certification of
Sec.~III.B cannot be implemented as a direct empirical block-success
analysis in the sense of Eq.~\eqref{eq:finite_product_bound}. We therefore
construct a conservative \emph{cycle-product} statistic from the
context-resolved data.

\paragraph*{Hardware implementation and readout mitigation:}
CHSH, Mermin--GHZ, and magic square were implemented using two, three,
and four measured qubits, respectively. For each hardware run, a
representative game circuit was first transpiled at optimization level
\(3\), and the resulting physical-qubit layout was then fixed for every
game and calibration circuit in that run. The game circuits were
transpiled at optimization level \(3\), whereas the computational-basis
calibration circuits were transpiled at optimization level \(1\). All
circuits were executed using IBM Runtime SamplerV2 without additional
Runtime resilience options.

Readout-error mitigation used a full assignment matrix on the measured
register. For a game involving \(n_q\) qubits, all \(2^{n_q}\)
computational-basis states were prepared and measured, producing a
\(2^{n_q}\times 2^{n_q}\) matrix
\begin{equation}
A_{ij}
=
P(\text{measured }i\mid\text{prepared }j).
\label{eq:assignment_matrix}
\end{equation}
Thus, the assignment matrices had dimensions \(4\times4\), \(8\times8\),
and \(16\times16\) for CHSH, Mermin--GHZ, and magic square,
respectively. The calibration circuits used the same fixed physical
qubits as the corresponding game circuits and were included in the same
hardware job. Given a raw count vector \(\mathbf{n}\), the corrected
vector was calculated using the Moore--Penrose pseudoinverse,
\begin{equation}
\widetilde{\mathbf{n}}
=
A^{+}\mathbf{n}.
\label{eq:assignment_matrix_inversion}
\end{equation}
Negative components produced by the inversion were set to zero, after
which the vector was renormalized to the original shot total.

Because the complete assignment matrix was used rather than a tensor
product of independent single-qubit matrices, this procedure can account
for static correlated assignment errors within the measured register.
It does not correct state-preparation or gate errors, nor does it model
time-dependent calibration drift or gate-induced crosstalk. No
zero-noise extrapolation, probabilistic error cancellation, dynamical
decoupling, or explicit drift compensation was applied. Placing the
calibration and game circuits in the same job reduces their temporal
separation, but does not eliminate possible drift within the job.
Accordingly, the difference between the raw and mitigated analyses is
reported as a sensitivity to the readout-correction model rather than as
a complete hardware-error correction.

\paragraph*{Simultaneous context-wise certification:}
For each primitive, let the one-round benchmark contexts be the complete
set of settings entering its game score: four contexts for CHSH, nine
for magic square, and four for Mermin--GHZ. A reconstructed cycle
contains one occurrence of every context. For raw data, let \(K_c\) and
\(N_c\) denote the integer number of wins and total shots in context
\(c\). We calculate an exact one-sided Clopper--Pearson lower confidence
limit
\begin{equation}
L_c(\alpha_c)
=
B^{-1}\!\left(
\alpha_c;K_c,N_c-K_c+1
\right)
\label{eq:context_cp_lower}
\end{equation}
for each context with \(K_c>0\), using \(L_c=0\) when \(K_c=0\).

To obtain simultaneous coverage without assuming statistical
independence among the context-wise estimates, we use the Bonferroni
allocation
\begin{equation}
\alpha_c=\frac{\alpha}{m},
\label{eq:bonferroni_context}
\end{equation}
where \(m\) is the number of contexts. The resulting simultaneous lower
bound on the reconstructed cycle-product score is
\begin{equation}
L_{\mathrm{block}}(\alpha)
:=
\prod_{c=1}^{m}
L_c\!\left(\frac{\alpha}{m}\right).
\label{eq:reconstructed_block_lower_bound}
\end{equation}
The corresponding benchmark-dependence certificate is
\begin{equation}
\eta_{\min}^{\mathrm{cert,block}}(\alpha)
=
\max\!\left\{
0,\,
L_{\mathrm{block}}(\alpha)^{1/m}
-\omega_{\mathrm c}
\right\}.
\label{eq:reconstructed_eta_block}
\end{equation}
Equation~\eqref{eq:reconstructed_eta_block} is used as a formal
finite-sample certificate only for the raw integer counts. Readout
mitigation produces fractional corrected frequencies that depend on the
finite calibration sample, matrix inversion, nonnegativity clipping, and
renormalization. These corrected frequencies are not binomial counts, so
applying a Clopper--Pearson formula to rounded or truncated mitigated
counts would not preserve its exact-coverage guarantee.

We therefore report the mitigated results separately as a sensitivity
analysis. Let \(\widetilde p_c\) denote the readout-corrected success
frequency for context \(c\). We define the mitigated reconstructed-block
plug-in quantities
\begin{equation}
\widetilde G_{\mathrm{block}}
:=
\prod_{c=1}^{m}\widetilde p_c,
\qquad
\widetilde{\eta}_{\min}^{\mathrm{mit,block}}
:=
\max\!\left\{
0,\,
\widetilde G_{\mathrm{block}}^{1/m}
-\omega_{\mathrm c}
\right\}.
\label{eq:mitigated_block_plugin}
\end{equation}
The tilde distinguishes this readout-corrected point estimate from the
confidence-certified raw-count quantity
\(\eta_{\min}^{\mathrm{cert,block}}\). A confidence-certified mitigated
analysis would additionally require propagation of both calibration and
game-shot uncertainty, for example through a calibration-aware
parametric bootstrap.

Table~\ref{tab:ibm_block_summary} summarizes the strongest exact
raw-count reconstructed-block certificate obtained for each primitive.
For CHSH, the largest certified value occurs on IBM Kingston with
4096 shots per context. The simultaneous one-sided lower bound is
\(
L_{\mathrm{block}}^{95\%}=0.4772,
\)
which gives
\(
\eta_{\min}^{\mathrm{cert,block}}=0.08115.
\)
This remains below the ideal CHSH threshold
\(\cos^2(\pi/8)-3/4\approx0.1036\), but establishes that any
benchmark-dependent classical surrogate reproducing the observed
raw-count cycle-product performance requires more than \(8.1\%\)
worst-case benchmark dependence at the stated confidence level.

For the magic-square game, every exact raw-count reconstructed-block
certificate is zero. The largest value of the simultaneous block lower
bound is obtained on IBM Fez with 128 shots per context,
\(
L_{\mathrm{block}}^{95\%}=0.02083,
\)
but its ninth root remains below the loophole-free classical value
\(\omega_{\mathrm c}=8/9\). This does not imply that the observed
one-round scores are classically trivial. Rather, the nine-context
product construction is sufficiently stringent that the simultaneous
finite-sample lower bound does not establish positive block-level
benchmark dependence.

Mermin--GHZ gives the strongest exact certificate. The largest raw-count
value is obtained on IBM Marrakesh with 2048 shots per context:
\(
L_{\mathrm{block}}^{95\%}=0.8772,
\qquad
\eta_{\min}^{\mathrm{cert,block}}=0.21777.
\)
Thus, with \(95\%\) confidence, any classical surrogate reproducing this
reconstructed cycle-product performance requires benchmark dependence
of at least \(21.8\%\). The certificate is below, but comparatively close
to, the ideal pseudo-telepathy threshold
\(1-\omega_{\mathrm c}=1/4\).

Readout mitigation systematically raises the corresponding plug-in
performance estimates for CHSH and Mermin--GHZ. The largest mitigated
value for CHSH is obtained on IBM Kingston with 4096 shots per context,
for which
\(
\widetilde{\eta}_{\min}^{\mathrm{mit,block}}=0.10190.
\)
For Mermin--GHZ, the largest mitigated value is obtained on IBM Fez with
2048 shots per context,
\(
\widetilde{\eta}_{\min}^{\mathrm{mit,block}}=0.24194.
\)
Across all backends and shot budgets, readout correction increases the
plug-in value by an average of \(8.12\times10^{-3}\) for CHSH and
\(1.49\times10^{-2}\) for Mermin--GHZ; the largest Mermin--GHZ increase
is \(2.78\times10^{-2}\). Magic square remains below the positive
block-threshold condition in both the raw and mitigated analyses. These
mitigated quantities quantify sensitivity to the readout model and
should not be interpreted as exact \(95\%\) confidence bounds.

The conclusions are insensitive to the particular simultaneous-interval
construction. For the strongest raw CHSH run, the Bonferroni
Clopper--Pearson, \v{S}idák Clopper--Pearson, Bonferroni Wilson, and
Bonferroni Agresti--Coull analyses give similar values $\simeq$ 0.0811. For the strongest raw Mermin--GHZ run, the corresponding values are $\simeq$ 0.217. All four constructions give zero for magic square. We retain the
Bonferroni Clopper--Pearson result as the primary certificate because it
does not require independence among context-wise estimates and retains
exact one-sided binomial coverage for the raw counts.

Figure~\ref{fig:finite_sample_block_certification} compares the exact
raw-count certificates with the readout-mitigated plug-in sensitivity
estimates over all backends and shot budgets. CHSH exhibits consistent
positive raw certification, while the corrected point estimates approach
the ideal CHSH threshold. Magic square remains uncertified under the
context-complete product reconstruction. Mermin--GHZ is the most robust
primitive under the present experimental conditions, combining a perfect
quantum target, a relatively low classical value, and only four benchmark
contexts.

These results should be interpreted together with the average-score analysis of Sec.~III.B. The average-score certification is statistically less stringent and therefore better aligned with noisy experimental implementations, while the reconstructed cycle-product certification reported here probes a stronger notion of loophole robustness. In particular, the IBM data already support nontrivial block-level certification for CHSH and especially for Mermin--GHZ, while also showing that the full magic-square product structure remains beyond reach under the present noise levels and finite sample sizes. From the standpoint of loophole-robust advantage claims, this separation is useful: it shows not only whether a primitive exceeds its loophole-free classical score on average, but also how far current hardware can sustain that advantage under a stricter finite-sample product-style benchmark.

\begin{table}[t]
\centering
\caption{Strongest exact reconstructed cycle-product certificates obtained
from the raw IBM hardware counts. Here \(m\) is the number of benchmark
contexts entering one reconstructed block,
\(L_{\mathrm{block}}^{95\%}\) is the simultaneous one-sided \(95\%\)
Clopper--Pearson lower bound obtained using Bonferroni allocation, and
\(\eta_{\min}^{\mathrm{cert,block}}\) is defined by
Eq.~\eqref{eq:reconstructed_eta_block}. Readout-mitigated results are
reported separately as plug-in sensitivity estimates because the
matrix-corrected frequencies are not binomial counts.}
\label{tab:ibm_block_summary}
\begin{tabular}{lccccc}
\hline
Primitive & Backend & Shots/context & \(m\) &
\(L_{\mathrm{block}}^{95\%}\) &
\(\eta_{\min}^{\mathrm{cert,block}}\) \\
\hline
CHSH & Kingston & 4096 & 4 & 0.4772 & 0.08115 \\
Magic square & Fez & 128 & 9 & 0.02083 & 0 \\
Mermin--GHZ & Marrakesh & 2048 & 4 & 0.8772 & 0.21777 \\
\hline
\end{tabular}
\end{table}

\subsection{Explicit benchmark-dependent classical surrogates}
\label{subsec:explicit_surrogates}

For the CHSH and Mermin--GHZ primitives, the additive relaxed bound
\(S_\eta \le \omega_c+\eta\) is not merely an abstract ceiling: it can be
saturated by an explicit benchmark-dependent classical model.

\paragraph*{CHSH:}
Let the benchmark instance be \(z=(x,y)\in\{00,01,10,11\}\), uniformly
distributed, with win condition \(a\oplus b=xy\). Consider four deterministic
strategies \(s_\lambda\), labeled by \(\lambda\in\{00,01,10,11\}\), each of
which loses on exactly one setting:
\begin{equation}
\begin{array}{c|c|c|c}
\lambda & a(x) & b(y) & \text{loses on} \\
\hline
11 & 0 & 0 & 11 \\
10 & x & 0 & 10 \\
01 & 0 & y & 01 \\
00 & 1-x & y & 00
\end{array}
\end{equation}
Let the hidden variable be \(\lambda\), with benchmark marginal
\(p_0(\lambda)=1/4\), and define
\begin{equation}
p(\lambda|z)=
\begin{cases}
\frac14-\eta, & \lambda=z,\\[1ex]
\frac14+\frac{\eta}{3}, & \lambda\neq z,
\end{cases}
\qquad 0\le \eta\le \frac14.
\end{equation}
Then \(d_{\mathrm{TV}}(p(\lambda|z),p_0)=\eta\) for every \(z\). Since the
strategy \(s_\lambda\) loses only when \(z=\lambda\), the conditional success
probability is \(
P(\mathrm{win}\mid z)=1-p(\lambda=z|z)=\frac34+\eta\). Hence the average score is \(S_\eta^{\mathrm{CHSH}}=\frac34+\eta.\)

\paragraph*{Mermin--GHZ:}
For the standard GHZ game with promised inputs
\(z\in\{000,011,101,110\}\), uniformly distributed, and winning condition
\begin{equation}
a\oplus b\oplus c =
\begin{cases}
0, & z=000,\\
1, & z\in\{011,101,110\},
\end{cases}
\end{equation}
consider four deterministic strategies \(s_\lambda\), each losing on exactly
one promised input:
\begin{equation}
\begin{array}{c|c|c|c|c}
\lambda & a(x) & b(y) & c(z) & \text{loses on} \\
\hline
000 & 1 & 0 & 0 & 000 \\
011 & x & 0 & 0 & 011 \\
101 & x & 0 & z & 101 \\
110 & x & y & 0 & 110
\end{array}
\end{equation}
With the same benchmark-dependent law
\begin{equation}
p(\lambda|z)=
\begin{cases}
\frac14-\eta, & \lambda=z,\\[1ex]
\frac14+\frac{\eta}{3}, & \lambda\neq z,
\end{cases}
\qquad 0\le \eta\le \frac14,
\end{equation}
one again has \(d_{\mathrm{TV}}(p(\lambda|z),p_0)=\eta\), and since each
strategy loses only on its labeled input,
\(
P(\mathrm{win}\mid z)=\frac34+\eta.
\)
Therefore
\(
S_\eta^{\mathrm{GHZ}}=\frac34+\eta.
\)

Thus, for both primitives, an explicit benchmark-dependent classical surrogate
exists whose performance increases linearly with \(\eta\) and exactly saturates
the relaxed average-score bound. In particular, for any observed average score
\(Q\), choosing \(\eta_{\mathrm{match}}=Q-\omega_c\) yields a classical
surrogate with \(S_{\eta_{\mathrm{match}}}=Q\), while choosing
\(\eta^\star=\eta_{\min}^{\mathrm{cert,avg}}+\delta\) with \(\delta>0\) yields
a surrogate whose score lies just above the certified lower-confidence quantum
performance.

\paragraph*{Magic square:}
The same construction extends to the Mermin--Peres magic-square game.
Let the benchmark instance be
\(z=(x,y)\in\{1,2,3\}\times\{1,2,3\}\), uniformly distributed over the
nine input pairs. The loophole-free classical value is
\(\omega_{\mathrm c}=8/9\). For each cell \((i,j)\), there exists a
deterministic classical strategy \(s_{ij}\) that loses only on the input
\((x,y)=(i,j)\) and wins on the remaining eight input pairs. Let the
hidden variable be \(\lambda=(i,j)\), with benchmark marginal
\(p_0(\lambda)=1/9\).

For \(0\leq\eta\leq 1/9\), define
\begin{equation}
p(\lambda|z)=
\begin{cases}
\displaystyle \frac{1}{9}-\eta,
& \lambda=z,\\[1ex]
\displaystyle \frac{1}{9}+\frac{\eta}{8},
& \lambda\neq z.
\end{cases}
\label{eq:magic_square_surrogate}
\end{equation}
The restriction \(\eta\leq 1/9\) ensures that all probabilities in
Eq.~\eqref{eq:magic_square_surrogate} are nonnegative. Moreover,
\begin{equation}
d_{\mathrm{TV}}\!\bigl(p(\lambda|z),p_0\bigr)
=
\frac{1}{2}
\left(
\eta+8\frac{\eta}{8}
\right)
=
\eta.
\label{eq:magic_square_surrogate_tv}
\end{equation}
Because \(s_\lambda\) loses only when \(z=\lambda\), its conditional
success probability is
\begin{equation}
P(\mathrm{win}|z)
=
1-p(\lambda=z|z)
=
\frac{8}{9}+\eta.
\label{eq:magic_square_surrogate_score}
\end{equation}
At \(\eta=1/9\), the strategy therefore wins with certainty. For every
larger allowed dependence \(\eta\geq 1/9\), the same perfect-score model
remains admissible because its actual benchmark dependence is only
\(1/9\). Consequently,
\begin{equation}
S_{\eta}^{\mathrm{MS}}
=
\min\!\left\{
1,\frac{8}{9}+\eta
\right\}.
\label{eq:magic_square_surrogate_optimum}
\end{equation}
Thus, the magic-square game also admits an explicit
benchmark-dependent classical surrogate that saturates the additive
relaxed bound over its full admissible range.

\subsection{Non-Bell benchmark-construction dependence in an ad hoc quantum-feature task}

The preceding examples are Bell-derived, but the relaxed advantage
inequality applies to any bounded-reward task.  We therefore next
consider a non-Bell supervised-learning benchmark based on the
quantum-enhanced feature-space construction of Havlíček \emph{et al.}
\cite{Havlicek2019}.  This example is useful for a different purpose
than the Bell-derived games.  It is not intended as a natural-data
demonstration of practical quantum advantage.  Instead, it provides a
controlled test of benchmark-construction dependence: the labels are
generated from the same quantum feature-space structure used by the
quantum kernel, so construction-side variables can be explicitly
recorded and audited.

From a classical machine-learning perspective, this audit is related to
tests for target leakage and shortcut learning
\cite{KaufmanEtAl2012,GeirhosEtAl2020}. The distinction is that the present
framework treats access to the unintended cue as an explicit classical
resource. Rather than only observing that a shortcut exists, we compare
the dependence carried by the candidate construction variable with the
minimum dependence required to explain the observed score gap, and then
test its operational usefulness using an explicit classical surrogate.

Each benchmark instance is a labeled example $z=(u,y)$,
where \(u\in[0,2\pi)^2\) is the official input feature vector and
\(y\in\{0,1\}\) is the class label used only for scoring.  The classifier outputs \(a=\hat y\), and the reward is the ordinary classification reward $r(z,a)=\mathbf 1\{a=y\}$.

Thus, the score is the held-out test accuracy, $S[p]=\mathbb E_{(u,y)\sim q}\left[\mathbf 1\{\hat y(u)=y\}\right]$, which is a bounded-reward task of the form considered in Sec.~II.

For the quantum classifier, we used a fidelity quantum kernel $K_Q(u,u') = \left|\langle \Phi(u)|\Phi(u')\rangle\right|^2$,
where $|\Phi(u)\rangle = U_\phi(u)|0\rangle^{\otimes 2}$
is generated by a two-qubit \(ZZ\)-feature map with two repetitions and
full entanglement.  On hardware, each kernel entry was estimated from
circuits of the form $U_\phi(u')^\dagger U_\phi(u)$, followed by computational-basis measurement.  The all-zero outcome
frequency estimates \(K_Q(u,u')\).  The diagonal training entries were
fixed to unity, consistent with \(K_Q(u,u)=1\), while all off-diagonal
and test-training entries were estimated from hardware counts.

The ad hoc labels were generated from a latent construction score
\begin{equation}
    g(u)=\langle\Phi(u)|O|\Phi(u)\rangle,
\end{equation}

where $O=V^{\dagger}\left(\prod_{j=1}^{n} Z_j\right)V$
is a fixed generator observable used only in constructing the benchmark. Here \(Z_j\) is the Pauli-\(Z\) operator on qubit \(j\), \(\prod_{j=1}^{n} Z_j\) is the computational-basis parity observable, and \(V\) is a fixed random unitary chosen during dataset generation. For the two-qubit task used here, this reduces to $O=V^{\dagger}(Z\otimes Z)V$. Thus, \(g(u)\) is the expectation value of a fixed rotated parity observable in the quantum feature state \(|\Phi(u)\rangle\). The classifier is given only the official input \(u\); \(O\), \(g(u)\), and quantities derived from \(g(u)\) are retained only for the benchmark-construction audit.

We used a separation gap
\(\Delta=0.25\): points with \(g(u)>\Delta\) were assigned to class
\(y=1\), points with \(g(u)<-\Delta\) were assigned to class \(y=0\), and
points satisfying \(|g(u)|\leq\Delta\) were rejected.  This procedure
creates a deliberately engineered task aligned with the quantum feature
map.  It also exposes natural construction-side candidate hidden
variables, including
\begin{equation}
\lambda_{\mathrm{sign}}=\mathrm{sign}\,g(u),
\end{equation}
a coarse signed-margin bin, a coarse absolute-margin bin, and a random
control variable.

For the finite held-out benchmark, let \(\widehat q\) denote the uniform
empirical distribution over the \(40\) test examples. For a categorical
construction-side variable \(\Lambda\), the direct empirical analogue of
the formal benchmark-dependence parameter in Eq.~\eqref{eq:eta_bm} would
be
\begin{equation}
\widehat{\eta}_{\mathrm{bm},\lambda}
=
\max_{(u,y)\in\mathcal T}
d_{\mathrm{TV}}
\!\left(
\widehat p(\lambda|u,y),
\widehat p_0(\lambda)
\right),
\label{eq:empirical_full_instance_eta}
\end{equation}
where \(\mathcal T\) is the held-out test set and
\(\widehat p_0(\lambda)\) is its empirical marginal.

This full-instance empirical quantity is not, however, an informative
diagnostic of class-relevant construction leakage in the present
benchmark. Each continuous input \(u\) occurs only once, and the
candidate variables are fixed functions or fixed assignments associated
with that input. Consequently,
\(\widehat p(\lambda|u,y)\) is a point mass. For a deterministic
categorical variable \(\lambda=h(u)\),
\begin{equation}
\widehat{\eta}_{\mathrm{bm},\lambda}
=
1-
\min_{\ell\in\operatorname{supp}\Lambda}
\widehat p_0(\ell),
\label{eq:deterministic_full_instance_eta}
\end{equation}
so the result is governed primarily by the least frequent category.
It can therefore assign substantial full-instance dependence even to a
random control that has negligible relation to the class label.

We instead report the label-conditioned construction diagnostic
\begin{equation}
\widehat{\eta}_{\lambda}^{(Y)}
=
\max_{y\in\{0,1\}}
d_{\mathrm{TV}}
\!\left(
\widehat p(\lambda|y),
\widehat p_0(\lambda)
\right).
\label{eq:label_conditioned_eta}
\end{equation}
This quantity measures the component of construction-side dependence
that is directly associated with the prediction target. It is not
identical to the formal full-instance parameter. Rather, it is a
coarse-grained lower bound on it. Indeed,
\begin{equation}
\widehat p(\lambda|y)
=
\sum_u
\widehat p(u|y)\,
\widehat p(\lambda|u,y),
\end{equation}
and convexity of total variation gives
\begin{align}
d_{\mathrm{TV}}
\!\left(
\widehat p(\lambda|y),
\widehat p_0(\lambda)
\right)
&\leq
\sum_u
\widehat p(u|y)\,
d_{\mathrm{TV}}
\!\left(
\widehat p(\lambda|u,y),
\widehat p_0(\lambda)
\right)
\nonumber\\
&\leq
\widehat{\eta}_{\mathrm{bm},\lambda}.
\end{align}
Therefore,
\begin{equation}
\widehat{\eta}_{\lambda}^{(Y)}
\leq
\widehat{\eta}_{\mathrm{bm},\lambda}.
\label{eq:coarse_eta_lower_bound}
\end{equation}
If
\(\widehat{\eta}_{\lambda}^{(Y)}>\eta_{\mathrm{req}}\), then the
candidate side resource necessarily has full benchmark dependence at
least as large as the amount required by the universal score bound.
This comparison alone is a necessary-resource audit, not a proof that
the variable can reproduce the quantum score. Operational sufficiency
is tested separately by training classical shortcut classifiers with
explicit access to \(\lambda\).

\begin{figure*}[t]
    \centering
    \includegraphics[width=0.95\textwidth]{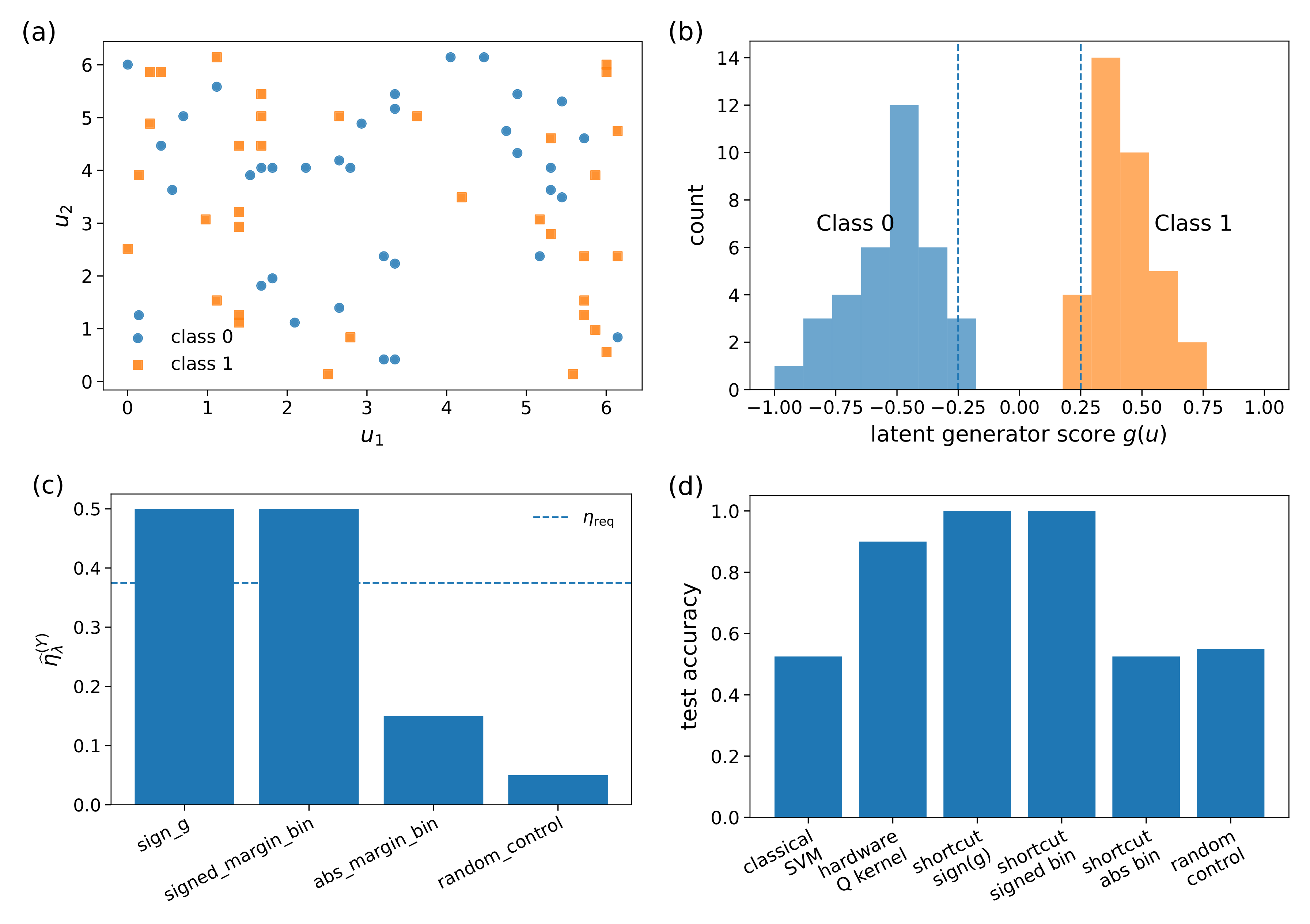}
    \caption{
    \textbf{Benchmark-construction dependence in a non-Bell ad hoc
    quantum-feature task.}
    \textbf{A,} Held-out and training examples from the engineered
    two-dimensional ad hoc benchmark, colored by class.
    \textbf{B,} Distribution of the latent generator score
    \(g(u)=\langle\Phi(u)|O|\Phi(u)\rangle\) by class.  The dashed
    vertical lines mark the separation gap \(\pm\Delta\), with
    \(\Delta=0.25\).  Labels are assigned by the sign of \(g(u)\) outside
    this rejected boundary region.
    \textbf{C,} Label-conditioned construction-side dependence
\(\widehat{\eta}_{\lambda}^{(Y)}\) for the candidate variables. This
coarse-grained quantity measures dependence associated with the class
label and lower-bounds the corresponding formal full-instance benchmark
dependence. The dashed horizontal line shows
\(\eta_{\mathrm{req}}=\max\{0,S_Q-S_{\mathrm{cl}}\}=0.375\).
The variables \(\mathrm{sign}\,g(u)\) and the signed-margin bin exceed
this required dependence.
    \textbf{D,} Held-out classification scores for the loophole-free
    classical SVM comparator, the hardware-estimated quantum kernel, and
    shortcut classifiers using construction-side variables.  Shortcut
    classifiers using \(\mathrm{sign}\,g(u)\) or the signed-margin bin
    reproduce or exceed the hardware quantum-kernel score.
    }
    \label{fig:adhoc_construction_dependence}
\end{figure*}
We implemented this task on the IBM
\texttt{ibm\_marrakesh} backend using \(1024\) shots per kernel circuit
(Fig.~\ref{fig:adhoc_construction_dependence}). The training set
contained \(15\) examples per class, and the held-out test set contained
\(20\) examples per class
(Figs.~\ref{fig:adhoc_construction_dependence}a and
\ref{fig:adhoc_construction_dependence}b).

The construction-side audit shows that the apparent excess is not
loophole-robust. Both \(\mathrm{sign}\,g(u)\) and the signed-margin bin
have
\[
\widehat{\eta}_{\lambda}^{(Y)}=0.5.
\]
This value follows directly from the label-generation rule. In the
retained ad hoc dataset, examples with \(g(u)>\Delta\) are assigned to
\(y=1\), whereas examples with \(g(u)<-\Delta\) are assigned to \(y=0\).
For
\(\lambda=\mathrm{sign}\,g(u)\), the class-conditional distribution is
therefore deterministic:
\begin{equation}
\widehat p(\lambda=+|y=1)=1,
\qquad
\widehat p(\lambda=-|y=1)=0,
\end{equation}
with the assignments reversed for \(y=0\). Because the held-out
benchmark is balanced,
\begin{equation}
\widehat p_0(\lambda=+)
=
\widehat p_0(\lambda=-)
=
\frac{1}{2}.
\end{equation}
Hence
\begin{align}
d_{\mathrm{TV}}
\!\left(
\widehat p(\lambda|y=1),
\widehat p_0(\lambda)
\right)
&=
\frac{1}{2}
\left(
\left|1-\frac{1}{2}\right|
+
\left|0-\frac{1}{2}\right|
\right)
\nonumber\\
&=
0.5.
\end{align}

The same reasoning applies to the signed-margin bin. Although it has
more than two categories, every retained negative-margin category occurs
only for \(y=0\), and every retained positive-margin category occurs
only for \(y=1\). It therefore retains the class label exactly and also
has
\(\widehat{\eta}_{\lambda}^{(Y)}=0.5\)
(Fig.~\ref{fig:adhoc_construction_dependence}c). As an average-case
cross-check, both signed variables satisfy
\begin{equation}
\widehat I(\Lambda;Y)
=
H(Y)
=
\log 2,
\end{equation}
or one bit, because the balanced binary label is a deterministic function
of either variable. By contrast, the absolute-margin bin has
\(\widehat{\eta}_{\lambda}^{(Y)}=0.15\), and the random control has only
\(\widehat{\eta}_{\lambda}^{(Y)}=0.05\). These latter quantities are
descriptive properties of the fixed finite held-out benchmark.
\begin{table*}[t]
\centering
\caption{
\textbf{Non-Bell ad hoc quantum-feature benchmark and construction-side
dependence audit.}
The task is a two-qubit Havlíček-style ad hoc classification benchmark
implemented with a hardware-estimated fidelity quantum kernel on
\texttt{ibm\_marrakesh}.  The official benchmark input is \(u\), while
construction-side variables \(\lambda\) are derived from the latent
generator score \(g(u)=\langle\Phi(u)|O|\Phi(u)\rangle\).  The score gap requires any benchmark-dependent classical explanation to
satisfy
\(\eta_{\mathrm{bm}}\geq
\eta_{\mathrm{req}}=
\max\{0,S_Q-S_{\mathrm{cl}}\}=0.375\).
The reported
\(\widehat{\eta}_{\lambda}^{(Y)}\) is the label-conditioned
construction diagnostic defined in
Eq.~\eqref{eq:label_conditioned_eta}; it is a lower bound on the formal
full-instance dependence of the candidate variable.
}
\label{tab:adhoc_construction_dependence}
\begin{tabular}{lcccc}
\hline
Model or construction variable & Input to classifier &
\(\widehat{\eta}_{\lambda}^{(Y)}\) &
Held-out score &
Interpretation \\
\hline
Classical SVM comparator & \(u\) only & -- & \(0.525\) &
Loophole-free classical comparator \\
Hardware quantum kernel & \(u\) only & -- & \(0.900\) &
Observed non-Bell quantum-kernel score \\
\(\mathrm{sign}\,g(u)\) & \(\lambda\) only & \(0.500\) & \(1.000\) &
Construction-side shortcut \\
Signed-margin bin & \(\lambda\) only & \(0.500\) & \(1.000\) &
Coarse construction-side shortcut \\
\(|g(u)|\) margin bin & \(\lambda\) only & \(0.150\) & \(0.525\) &
Boundary-strength variable only \\
Random control & \(\lambda\) only & \(0.050\) & \(0.550\) &
Negative control \\
\hline
\end{tabular}
\end{table*}

The hardware-estimated quantum-kernel support vector machine (SVM) achieved $S_Q=0.900$, whereas the best pre-specified classical SVM comparator, selected from
linear, polynomial, and Gaussian radial basis function (RBF) kernels using training data only,
was a polynomial-kernel SVM with $S_{\mathrm{cl}}=0.525$ (Fig. ~\ref{fig:adhoc_construction_dependence}d).
The apparent non-Bell excess was therefore \(
S_Q-S_{\mathrm{cl}}=0.375.\)
By the universal relaxed inequality, any benchmark-dependent classical
surrogate reproducing the hardware quantum-kernel score must have
\begin{equation}
\eta_{\mathrm{bm}}
\geq
\eta_{\mathrm{req}}
:=
\max\{0,S_Q-S_{\mathrm{cl}}\}
=
0.375.
\label{eq:qml_required_dependence}
\end{equation}

The signed construction variables satisfy
\(
\widehat{\eta}_{\lambda}^{(Y)}
=
0.5
>
\eta_{\mathrm{req}},
\)
and Eq.~\eqref{eq:coarse_eta_lower_bound} implies that their corresponding
full-instance benchmark dependence is also at least \(0.5\). They
therefore possess more than the minimum dependence budget required by
the universal bound. Whether that dependence is operationally useful is
determined by the shortcut-classifier analysis below.

The paired held-out predictions show that the quantum kernel was correct
on \(36/40\) test examples, while the classical SVM was correct on
\(21/40\).  On the same test examples, the quantum kernel was correct and
the classical SVM was wrong in \(17\) cases, while the classical SVM was
correct and the quantum kernel was wrong in \(2\) cases.  The paired
score difference is 0.375, which is consistent with the direct score difference above.

Following the construction-side audit, we also trained classical SVM surrogates with explicit access to the
construction-side variables.  This is not a fair classical benchmark;
rather, it is a loophole-enabled classical model designed to test
whether side information can operationally reproduce the quantum-kernel
score.  A classical SVM given \(\mathrm{sign}\,g(u)\) alone achieved
perfect test accuracy.  The same was true for a classical SVM given the
continuous latent generator score \(g(u)\).  Thus, in this engineered
task,
Let \(S_{\mathrm{cl}}^{(\lambda)}\) denote the held-out score of a
classical surrogate given explicit access to the construction-side
variable \(\lambda\). For the signed variables,
\begin{equation}
S_{\mathrm{cl}}
<
S_Q
<
S_{\mathrm{cl}}^{(\lambda)}
=
1.
\label{eq:qml_shortcut_ordering}
\end{equation}
for the relevant construction-side variables.  This is the operational
signature of benchmark-construction dependence: the quantum kernel
outperforms the loophole-free classical comparator on the official input,
but a classical surrogate with access to construction-side information
can reproduce or exceed the quantum-kernel performance.
A summary of ad hoc quantum-feature benchmark and construction-side dependence audit results are
shown in Table~\ref{tab:adhoc_construction_dependence}. The latent generator score separates the two classes by construction.
The signed construction-side variables have
\(\widehat{\eta}_{\lambda}^{(Y)}>\eta_{\mathrm{req}}\), and the
corresponding shortcut classifiers outperform both the loophole-free
classical comparator and the hardware quantum-kernel classifier.
Together, the dependence audit and the explicit shortcut performance
show that the apparent quantum-kernel excess is explainable by
benchmark-construction side information.

This experiment should therefore be interpreted as a benchmark audit,
not as a claim of practical quantum advantage.  The ad hoc task is
intentionally favorable to the quantum feature map, and the same
construction introduces latent variables that are strongly correlated
with the benchmark labels.  The value of the present framework is that it
makes this fragility quantitative.  The observed hardware quantum-kernel
excess requires \(\eta_{\mathrm{req}}=0.375\) benchmark dependence to
explain, while the signed construction-side variables have
\(\widehat{\eta}_{\lambda}=0.5\) and enable classical surrogates with
perfect test accuracy.  Thus, even a positive non-Bell quantum-kernel
result can fail to be loophole-robust once benchmark-construction side
resources are explicitly audited.

\section{Conclusion}

We introduced benchmark dependence as a task-level resource for evaluating
the loophole robustness of claimed quantum advantage. For a hidden state
whose conditional law may depend on the evaluated benchmark instance, the
optimal classical score of any bounded-reward task satisfies
\[
S_\eta
\leq
\min\!\left\{
1,S_{\mathrm{cl}}+\eta
\right\}.
\]
The total-variation stability estimate entering the proof is classical;
the contribution of the present framework is its operational use for
quantum-advantage certification. It identifies benchmark-correlated side
information as an explicit classical resource, optimizes over the resulting
strategy class, and converts an observed score gap into the minimum
dependence required by a classical explanation. The hidden-instance
construction proves that the unit coefficient of \(\eta\) is universally
optimal.

Additional task structure permits sharper conclusions. For repeated
product tasks under roundwise benchmark dependence, we obtained
\[
S_\eta^{(n)}
\leq
(\omega_{\mathrm c}+\eta)^n,
\]
with exact saturation for an explicit family. We also treated
average-of-rounds scores, finite samples, multipartite tasks,
mutual-information dependence, and correlated hidden trajectories. The
multiplicative result requires roundwise factorization; when arbitrary
inter-round correlations are allowed, the pathwise theorem supplies a
more conservative but generally valid alternative.

The hardware analysis illustrates the distinction between formal
certification and error-mitigation sensitivity. Using simultaneous
one-sided Clopper--Pearson limits on the raw integer counts, the strongest
reconstructed cycle-product certificates are
\(\eta_{\min}^{\mathrm{cert,block}}=0.08115\) for CHSH on IBM Kingston
and \(0.21777\) for Mermin--GHZ on IBM Marrakesh. The magic-square
cycle-product remains uncertified because its nine-context product is much
more demanding. Assignment-matrix mitigation raises the corresponding
point estimates, but the corrected frequencies are not binomial counts;
we therefore report those results as readout-model sensitivity estimates
rather than exact-confidence certificates.

The non-Bell quantum-kernel example shows why the resource model matters
outside Bell tests. The hardware quantum kernel exceeds the pre-specified
classical comparator on the official input, but the benchmark labels are
generated from a latent quantum-feature-space score. Signed variables
derived from that score have
\(\widehat{\eta}_{\lambda}^{(Y)}=0.5>\eta_{\mathrm{req}}=0.375\) and
support classical shortcut classifiers with perfect held-out accuracy.
This example is not a claim of practical quantum advantage; it is a
controlled demonstration that an apparent advantage can disappear when
benchmark-construction side information is audited explicitly.

The resulting certification principle is therefore stronger than the
statement that a selected classical comparator has been exceeded. A
loophole-robust claim should specify the permitted information resources,
report the benchmark-independent classical value, and quantify the
minimum benchmark dependence required by any classical surrogate capable
of reproducing the observed score. This extends the operational logic of
relaxed Bell tests to distributed quantum tasks and quantum-machine-learning
benchmarks.

\section*{Data Availability Statement}

The data and analysis code supporting the findings of this study, including
the raw and readout-mitigated IBM quantum-hardware results and the ad hoc
quantum-kernel benchmark data, are publicly available at
\url{https://github.com/ramkrip/quantadvantage/}.

\bibliographystyle{apsrev4-2}  % PRR/APS-compatible style
\bibliography{refs}   

\end{document}